\documentclass[12pt]{article}

\usepackage{amsmath,amstext,amssymb,amsthm,graphicx,appendix}
\usepackage{verbatim}
\usepackage{color}
\usepackage{algorithm,algorithmic}
\usepackage{authblk}
\hsize \textwidth
\advance \hsize by -\marginparwidth
\evensidemargin \oddsidemargin

\advance\hoffset by 5mm

\def\Ac{{\cal A}}

\def\Gc{{\cal G}}
\def\Ec{{\cal E}}
\def\Rc{{\cal R}}
\def\Zc{{\cal Z}}

\def\vect#1{\mbox{\boldmath{$#1$}}}

\def\mR{\mathbb{R}}
\def\mC{\mathbb{C}}

\def\wG{\widehat{G}}
\def\wU{\widehat{U}}
\def\wV{\widehat{V}}
\def\wf{\widehat{f}}
\def\wg{\widehat{g}}

\def\wP{\widehat P}

\def\hxi{\hat\psi}
\def\Rho{{X}}

\def\sinc{\operatorname{sinc}}
\def\diag{\operatorname{diag}}

\newcommand{\argmin}{\mathop{\mbox{argmin}}}
\newcommand{\ltr}{\left\langle}
\newcommand{\rtr}{\right\rangle}

\def\bfrho{\mbox{\boldmath$\rho$}}
\def\bfgamma{\mbox{\boldmath$\gamma$}}

\newtheorem{thm}{Theorem}[section]
\newtheorem{rem}[thm]{Remark}
\newtheorem{prop}[thm]{Proposition}
\newtheorem{lem}[thm]{Lemma}
\newtheorem{col}[thm]{Corollary}

\newenvironment{definition}[1][Definition]{\begin{trivlist}
\item[\hskip \labelsep {\bfseries #1}]}{\end{trivlist}}

\title{Imaging strong localized scatterers with sparsity promoting optimization\footnote{This version: August 1, 2013. M. M. and G. P. were supported by AFOSR grant FA9550-11-1-0266.}} 
\author[1]{Anwei Chai\thanks{anwei@math.stanford.edu}}
\author[2]{Miguel Moscoso\thanks{moscoso@math.uc3m.es}}
\author[1]{George Papanicolaou\thanks{papanico@math.stanford.edu}}
\affil[1]{Department of Mathematics, Stanford University, California 94305, USA}
\affil[2]{Gregorio Mill\'{a}n Institute, Universidad Carlos III de Madrid, Madrid 28911, Spain}
\date{}
\begin{document}
\maketitle

\begin{abstract}
We study active array imaging of small but strong scatterers in homogeneous media when multiple scattering
between them is important. We use the Foldy-Lax equations to model wave propagation with multiple scattering when
the scatterers are small relative to the wavelength. 
In active array imaging we seek to locate the positions and reflectivities of the scatterers, that is, to determine the support of the reflectivity
vector and the values of its nonzero elements from echoes recorded on the array.
This is a nonlinear inverse problem because of
the multiple scattering. We show in this paper how to avoid the nonlinearity and form images non-iteratively through a two-step process which involves
$\ell_1$ norm minimization. However, under certain illuminations imaging may be
affected by screening, where some scatterers are obscured by multiple scattering.
This problem can be mitigated by using multiple and diverse illuminations.
In this case, we determine 
solution vectors that have a common support. The uniqueness and stability of the support of the reflectivity vector
obtained with single or multiple illuminations 
are analyzed, showing that the errors are proportional to the 
amount of noise in the data with a proportionality factor 
dependent on the sparsity of the solution and the mutual coherence of the sensing matrix, which is determined by the geometry of the imaging array.
Finally, to filter out noise and improve the resolution of the images, we propose an approach that combines optimal illuminations using the
singular value decomposition of the response matrix together with  sparsity promoting optimization
jointly for all illuminations.
This work is an extension of our previous paper \cite{CMP13} on imaging using optimization
techniques where we now account for multiple scattering effects.
\end{abstract}


\section{Introduction}

Active array imaging when multiple scattering between the scatterers is important
is challenging because it is a nonlinear inverse problem. In most applications, for example, in seismic imaging, ultrasonic
non-destructive testing, synthetic aperture radar, etc., the imaging methods that are used ignore multiple scattering and
deal with a linear inverse problem. This may result in some loss of resolution, especially if the imaging setup provides
only partial information with, for example, a limited frequency range, limited illuminations or small arrays relative
to the distance from the scatterers. 
In this paper, we study active array imaging at only one frequency and with single and multiple illuminations. The arrays
considered are not small and could surround the scatterers. We consider the full nonlinear inverse problem when multiple
scattering is included but limit the analysis to the case of imaging when the scatterers are small compared
to the wavelength so that the Foldy-Lax approximation \cite{F45,L51,L52,BTPB02,martin06} can be used.
Given the array data, we formulate imaging as
an underdetermined optimization problem with nonlinear constraints. It is underdetermined because the set of possible locations
of the scatterers in the image regions we consider here is much larger than the array data set, as
is often the case.

In this paper, we formulate the nonlinear optimization problem for imaging in two steps.
In the first step, we treat the scatterers as equivalent sources of unknown locations whose strengths
are also unknown but are related in a known way to the illumination, to the multiple scattering and to the underlying unknown
reflectivities of the scatterers.
Under specific, if somewhat conservative, hypotheses about the array imaging setup and the
measurement noise level, we show that the location of the sources, that is, the scatterers, 
can be recovered exactly in the first step. This is because we employ an $\ell_1$ minimization method
that recovers the support of sparse solutions exactly.  In the second step, once the location of the scatterers is fixed,
their true reflectivities can be recovered using the known relationship to the source strengths obtained in step one.
This is an explicit relation that comes from the Foldy-Lax equations, given the scatterer locations. 
The key to this two-step approach is the possibility of recovering exactly the locations of the
scatterers in the first step. This effectively linearizes what is a nonlinear inverse problem.
Our theoretical analysis is mainly based on the work by Tropp \cite{TROPP04}. We give sufficient 
conditions on the imaging setup under which we can expect perfect reconstructions without noise, and conditions under which these
reconstructions are stable when the data is contaminated with additive noise.

We note that the two-step imaging method we have described is applied at first to array
data from a single illumination, in \S\ref{sec:SMV}. However, because of the {\em screening effects},
not all scatterers can be recovered from data generated by a single illumination, in general.
Moreover,
using a single illumination with array imaging configurations often used in practice 
is not robust when data is contaminated with noise.   These issues can be 
handled by applying the same two-step imaging method to data from multiple illuminations. For this case, a
matrix version of the $\ell_1$ minimization method is described and analyzed in \S\ref{sec:MMV}. 

When we have access to the full array response matrix, that is, when we have data for a full set of linearly independent illumination
vectors, it is possible to image using its singular value decomposition (SVD).
This is done in \cite{Gruber04,DMG05,MG06} where the authors show that, despite of the nonlinearity of the mapping from the reflectivities of the scatterers to the
data received on the array, one can form an image and find the locations of the individual point scatterers by beamforming with the significant singular vectors. This
is the MUSIC (MUltiple SIgnal Classification) method, which is essentially the same as beamforming or migration using the significant singular
vectors as illuminations. These illumination vectors are optimal in the sense that they result in array data with maximal power, which is proportional to the
associated singular values.

We also use optimal illuminations in the $\ell_1$ based minimization approach we introduce here.
This increases the efficiency of the data collection process and lowers the impact of the noise in the data, 
as we show with numerical simulations in \S\ref{simulation}.
This is so even when only a few optimal illuminations are used. We used
optimal illuminations in \cite{CMP13} for a proposed hybrid $\ell_1$ method in the context of array
imaging when multiple scattering is negligible. The hybrid $\ell_1$ method not only uses the optimal
illumination but also projects the data into appropriate subspaces. This last step cannot be carried out when multiple
scattering is important because after projection the sensing matrix contains unknown factors related to the multiple scattering.

Finally,  the performance of the imaging method we propose 
depend on the mutual coherence of the sensing matrix, defined in \S\ref{sec:SMV}.
We, therefore, carry out an analysis of the mutual coherence of the sensing matrix upon different imaging configurations,
with planar and spherical arrays, in \S\ref{sec:arrayconfig}. We show that spherical arrays give images with better resolutions
and smaller upper bound of the corresponding sensing matrix than planar arrays. 

We note that the formulation for imaging with non-negligible multiple scattering can also
be combined with the optimization approach for imaging problems with intensity only measurements, which is 
studied by the authors in \cite{CMP11}. 

The organization of the paper is as follows. In the rest of this section we introduce some basic
notations used throughout the paper. In \S\ref{sec:formulation}, we formulate the array imaging problem with 
multiple scattering between the scatterers using the Foldy-Lax equations. In \S\ref{sec:SMV}, we analyze the imaging problem with a
single illumination, and with and without additive noise. In \S\ref{sec:MMV}, we discuss active array imaging with multiple illuminations, where we also discuss the use of optimal illuminations and give an
efficient algorithm for solving multiple measurement vector (MMV) problems. In \S\ref{sec:arrayconfig}, we discuss the 
impact of the array configuration on the performance of the methods proposed in the paper.
In \S\ref{simulation}, we illustrate our results with various numerical examples under conditions much less conservative
than the ones required by the theory. The proofs of all
the theoretical results are given in the appendices at the end of the paper.

\subsection*{Notation}
Throughout the paper, vectors are denoted by lower case letters in boldface, and matrices by capital letters in boldface.
Given a vector $\vect v$ of length $N$, the $i^\mathrm{th}$ entry is denoted by $v_i$. For a matrix $\mathbf{Y}$ of
size $M\times N$, the $i^\mathrm{th}$ row is denoted by $Y_{i\cdot}$, the $j^\mathrm{th}$ column 
by $Y_{\cdot j}$, and the $(i,j)$ entry by $Y_{ij}$. 

We will use several different norms. For $p\ge1$, $q\ge1$, the $\ell_p$ norm of vector $\vect v$ is
defined by \[\|\vect v\|_{\ell_p}=\left(\sum_{i=1}^n|v_i|^p\right)^{1/p},\]
while the Frobenius norm of a matrix $\mathbf{Y}$ is given by
\[\|\mathbf{Y}\|_F=\left(\sum_{i=1}^m\sum_{j=1}^n|Y_{ij}|^2\right)^{1/2}=\sqrt{\operatorname{trace}(\mathbf{Y} \mathbf{Y}^\ast)}.\]
Here, $\ast$ is the conjugate transpose operator of vectors and matrices. We will use the superscript $T$ for the transpose only operator.
We will also need to use the operator norm of a matrix, defined as
\[
\|\mathbf{Y}\|_{p\rightarrow q}=\max_{\vect v\neq0}\frac{\|\mathbf{Y}\vect v\|_{\ell_q}}{\|\vect v\|_{\ell_p}},
\]
and the $(p,q)$-norm function $J_{p,q}(\cdot)$ defined as
\begin{equation}
J_{p,q}(\mathbf{Y})=\left(\sum_{i=1}^m\|Y_{i\cdot}\|_{\ell_p}^q\right)^{1/q},
\label{eq:Jpq}
\end{equation}
which is simply the $\ell_q$ norm of the vector formed by $\ell_p$ norm of all rows of a matrix.

\section{Formulation of active array imaging}\label{sec:formulation}
In active array imaging we seek to locate the positions and reflectivities of a set of scatterers using the data recorded on an array ${\cal A}$. 
By an active array, we mean a collection of $N$ transducers that emit spherical wave signals from positions $\vect x_s\in{\cal A}$ and record the echoes
with receivers at positions $\vect x_r\in{\cal A}$. The transducers are placed at distance $h$ between them, which is of
the order of the wavelength $\lambda=2\pi c_0/\omega$, where $c_0$ is the wave speed in the medium and $\omega$ is the frequency of the probing signal.

We now introduce the direct and inverse scattering problems for imaging 
point-like scatterers with an active array in a homogeneous medium. We 
consider the case in which multiple scattering among $M$ scatterers is not negligible.
The scatterers, with unknown reflectivities $\alpha_j\in\mC$ and positions $\vect y_{n_j}$, $j=1,\ldots,M$,
are assumed to be located within a region of interest called the image window (IW), which is centered at a distance
$L$ from the array. We discretize the IW using a uniform grid of $K$ points $\vect y_j$, $j=1,\ldots,K$, 
and assume that each scatterer is located at one of these $K$ grid points so that
$$\{\vect y_{n_1},\ldots,\vect y_{n_M}\}\subset\{\vect y_1,\ldots,\vect y_K\}.$$
Furthermore, we assume that near-field multiple scattering effects are negligible because the scatterers are sufficiently far apart.
Thus, we can use the far-field approximation 
\begin{equation}\label{greenfunc}
\wG_0(\vect y,\vect x,\omega)=\frac{\exp(i \kappa|\vect x-\vect y|)}{4\pi|\vect x-\vect y|}\, ,
\end{equation}
with $\kappa=\omega/c_0$, for the free-space Green's function $\wG_0(\vect y,\vect x,\omega)$ to characterize
wave propagation from point $\vect x$ to point $\vect y$ in the homogeneous medium. 

We formulate the scattered wave field in a homogeneous medium in terms of the Foldy-Lax equations \cite{F45,L51,L52}.
In this model framework, the scattered wave received at transducer $\vect x_r$ due to a narrow band signal of angular frequency $\omega$
sent from $\vect x_s$ can be written as the sum of all scattered waves from the $M$ scatterers
\begin{equation}\label{eq:P}
\wP(\vect x_r, \vect x_s)=\sum_{j=1}^M\hxi_{j}^s(\vect x_r; \vect y_{n_1},\ldots,\vect y_{n_M}).
\end{equation}
Here, and in all that follows, we will drop the dependence of waves and measurements on the frequency $\omega$.
In \eqref{eq:P}, $\hxi_{j}^s(\vect x_r;\vect y_{n_1},\ldots,\vect y_{n_M})$ represents the scattered wave observed at
$\vect x_r$ due to the scatterer at position $\vect y_{n_j}$.
It actually depends on the positions of all the scatterers
$\vect y_{n_j}$, $j=1,\ldots,M$, and it is given by
\begin{equation}\label{eq:scattered-field}
\hxi_{j}^s(\vect x_r; \vect y_{n_1},\ldots,\vect y_{n_M})=\alpha_j\wG_0(\vect x_r,\vect y_{n_j})\hxi_{j}^e(\vect y_{n_1},\ldots,\vect y_{n_M}),
\end{equation}
where $\hxi_{j}^e(\vect y_{n_1},\ldots,\vect y_{n_M})$ represents the exciting field at the scatterer located at $\vect y_{n_j}$.
Because we ignore self-interacting fields, the exciting field at $\vect y_{n_j}$ is equal to the sum of the incident field
$\hxi^{inc}_j:=\hxi^{inc}(\vect y_{n_j},\vect x_s)$ at $\vect y_{n_j}$ and the scattered fields at $\vect y_{n_j}$ due to all scatterers except for the one at $\vect y_{n_j}$. Hence, it is given by 
\begin{equation}
\hxi_{j}^e(\vect y_{n_1},\ldots,\vect y_{n_M})= \hxi^{inc}(\vect y_{n_j},\vect x_s)  +
\sum_{k\neq j}\alpha_k\wG_0(\vect y_{n_j},\vect y_{n_k})\hxi_{k}^e(\vect y_{n_1},\ldots,\vect y_{n_M}),\quad j=1,2,\ldots,M\, .
\label{eq:effectivefields}
\end{equation}
This is a self-consistent system of $M$ equations for the $M$ unknown exciting fields
$$\hxi_1^e:=\hxi_{1}^e(\vect y_{n_1},\ldots,\vect y_{n_M}),\ldots, \hxi_M^e:=\hxi_{M}^e(\vect y_{n_1},\ldots,\vect y_{n_M}), $$
which  can be written in matrix form as
\begin{equation}\label{eq:FL}
\mathbf{Z}_M (\vect\alpha) \, \mathbf{\Phi}^e=\mathbf{\Phi}^{inc} \,\, .
\end{equation}
In \eqref{eq:FL}, $\mathbf{\Phi}^e = [\hxi^e_1,\ldots,\hxi^e_M]^T$ and
$\mathbf{\Phi}^{inc} =[\hxi^{inc}_1,\ldots,\hxi^{inc}_M]^T$ are 
vectors whose components are the exciting and incident fields on the $M$ scatterers, respectively, and
\begin{equation}
\label{eq:Zm}
\big(Z_M (\vect\alpha)\big)_{ij}=\begin{cases}1,& i=j\\
-\alpha_j\wG_0(\vect y_{n_i},\vect y_{n_j}),& i\neq j\, ,
\end{cases}
\end{equation}
is the $M\times M$ Foldy-Lax matrix which depends on  the reflectivities $\vect\alpha=(\alpha_1,\ldots,\alpha_M)$.
With the solution of \eqref{eq:FL}, we  use \eqref{eq:scattered-field} and
\eqref{eq:P} to compute the scattered data received at the array. 

Note that the exciting fields $\mathbf{\Phi}^e$ depend on the incident fields $\mathbf{\Phi}^{inc}$ and, hence, they depend on the illumination sent from the array.
To characterize it, we define the illumination vector $\vect\wf=[\wf_1,\ldots,\wf_N]^T$ whose entries denote the strength of the signals sent 
from each of the $N$ transducers in the array. We will assume that the illumination vectors are normalized, so $\|\vect\wf\|_{\ell_2}=1$.

To write the data received on the array in a more compact form, we define the Green's function or steering vector $\vect\wg_0(\vect y)$ 
at location $\vect y$ in the IW as 
\begin{equation}\label{GreenFuncVec}
\vect \wg_0(\vect y)=[\wG_0(\vect x_{1},\vect y), \wG_0(\vect x_{2},\vect y),\ldots,
\wG_0(\vect x_{N},\vect y)]^T.
\end{equation}
Then, given any illumination vector $\vect\wf$, the incident field on the scatterer at position $\vect y_{n_j}$  is equal to
$\vect\wg_0^T(\vect y_{n_j})\vect\wf$. If the illumination vector $\vect\wf$ is such that $\wf_s=1$ and $\wf_j=0$ for $j=1,\ldots,N$ with $j\neq s$, then
 the incident field at $\vect y_{n_j}$ is simply $\wG_0(\vect y_{n_j}, \vect x_s)$.

Using \eqref{GreenFuncVec}, we also define the $N\times K$ sensing matrix $\vect{\Gc}$ as
\begin{equation}\label{sensingmatrix}
\vect{\Gc}=[\vect\wg_0(\vect y_1)\,\cdots\,\vect\wg_0(\vect y_K)] \,,
\end{equation}
and the $N\times M$ submatrix corresponding to the locations of scatterers as
$$\vect{\Gc}_M=[\vect\wg_0(\vect y_{n_1})\,\cdots\,\vect\wg_0(\vect y_{n_M})].$$
With this notation, the array response matrix can be written as
\begin{equation}\label{responsematrix1}
\vect\wP\equiv[\wP(\vect x_r,\vect x_s)]_{r,s=1}^N=\vect\Gc_M\diag(\vect\alpha)\vect Z_M^{-1}(\vect\alpha)\vect\Gc^T_M,
\end{equation}
and the data received on the array due to the illumination $\vect\wf$ is
\begin{equation}\label{data}
\vect b=\vect\wP\vect\wf. 
\end{equation}
Note that the response matrix in \eqref{responsematrix1} that takes into account multiple scattering,
includes the inverse of the Foldy-Lax matrix $\mathbf{Z}_M^{-1}(\vect\alpha)$. When multiple scattering is
negligible, $\mathbf{Z}_M(\vect\alpha)=\mathbf{I}$ and we get the response matrix under the Born approximation, as shown 
for example in \cite{CMP13}.
We further note that the response matrix $\vect\wP$ given by \eqref{responsematrix1} is  symmetric.

Next, we introduce the true {\it reflectivity vector}
$\vect\rho_0=[\rho_{01},\ldots,\rho_{0K}]^T\in\mC^K$ such that
$$\rho_{0k}=\sum_{j=1}^M\alpha_j\delta_{\vect y_{n_j}\vect y_k},\,\, k=1,\ldots,K,$$ where
$\delta_{\cdot\cdot}$ is the classical Kronecker delta.
Note that the Foldy-Lax matrix $\mathbf{Z}_M(\vect\alpha)$ is defined only for pairwise combinations of
scatterers at $\vect y_{n_j}$, $j=1,\ldots,M$. To formulate the inverse scattering problem,
we need to extend the $M\times M$ matrix $\mathbf{Z}_M(\vect\alpha)$ to a larger $K\times K$ matrix 
\begin{equation}
\label{eq:Z}
\big(Z(\vect\rho_0)\big)_{ij}=\begin{cases}1,& i=j\\
-\rho_{0j}\wG_0(\vect y_{i},\vect y_{j}),& i\neq j\, ,
\end{cases}
\end{equation}
which includes all pairwise combinations of the $K$ grid points $\vect y_j$
in the IW. 
With this notation, the array response matrix \eqref{responsematrix1} can be written as
\begin{equation}\label{responsematrix2}
\vect\wP=\vect\Gc\diag(\bfrho_0)\mathbf{Z}^{-1}(\bfrho_0)\vect\Gc^T.
\end{equation}
Furthermore, if we define the Foldy-Lax Green's function
vector $\vect\wg_{FL}(\vect y_j)$, $j=1,\ldots,K$, as the $j^\mathrm{th}$ column of the matrix
 $\vect\Gc_{FL}(\vect\rho)=\vect\Gc\mathbf{Z}^{-T}(\vect\rho)$, i.e.,
 \begin{equation}\label{eq:gFL}
 \begin{bmatrix}
 \vect\wg_{FL}(\vect y_1)&\cdots&\vect\wg_{FL}(\vect y_K)
 \end{bmatrix}=\vect\Gc\mathbf{Z}^{-T}(\vect\rho),
 \end{equation}
 then \eqref{responsematrix2} can be simplified to 
 \begin{equation}\label{responsematrix3}
\vect\wP=\vect\Gc\diag(\bfrho_0)\vect\Gc_{FL}^T(\vect\rho_0).
\end{equation}

Given an illumination vector $\vect\wf$ and the configuration of scatterers in the IW characterized by $\vect\rho_0$, the data received on the array is given by \eqref{data}.
The array imaging problem when a single illumination is used to probe the medium is to find the true reflectivity vector $\vect\rho_0$ from the
received data $\vect b$. The detailed formulation of this problem will be discussed in depth in \S\ref{sec:SMV}. The array imaging problem
that uses a collection of array data generated by different illumination vectors will be discussed in \S\ref{sec:MMV}. In either situation,
our method for active array imaging with multiple scattering is noniterative. It uses two steps to get the images: first locating the
scatterers and second computing their reflectivities.

\section{Active array imaging with single illumination}\label{sec:SMV}
In this section, we show the formulation of active array imaging including multiple scattering when only one illumination
is sent from the array to probe the medium. In this case,  a single measurement vector is used to infer the location
and reflectivities of the scatterers. In signal processing literature, this problem belongs to the so called
{\it Single Measurement Vector} (SMV) problem.

For a given illumination vector $\vect\wf$, we define the operator $\vect\Ac_{\wf}$ through the identity
\begin{equation*}
\vect\Ac_{\wf}\bfrho_0=\vect\wP\vect\wf,
\end{equation*}
which connects the reflectivity vector $\bfrho_0$ and the data \eqref{data}.
It is easy to see from \eqref{responsematrix2} that $\vect\Ac_{\wf}$ has the form
\[\vect\Ac_{\wf}=[\wg_{\wf}(\vect y_1)\vect\wg_0(\vect y_1)\,\cdots\,\wg_{\wf}(\vect y_K)\vect\wg_0(\vect y_K)],\]
where $\wg_{\wf}(\vect y_j)=\vect\wg^T_{FL}(\vect y_j)\vect\wf$, $j=1,\dots,K$, are scalars. With this notation,
active array imaging with a single illumination amounts to solving $\bfrho_0$ from the system of equations 
\begin{equation}\label{eq:single}
\vect\Ac_{\wf}\bfrho=\vect b.
\end{equation}
The number of transducers $N$ is usually much smaller than the number of the grid points $K$ in the IW and, hence, \eqref{eq:single} is an underdetermined system of equations. 

Although equations \eqref{eq:single} are exactly of the same form as the problem studied
in \cite{CMP13}, there is a substantial difference. Due to the multiple scattering among the scatterers,
the terms $\wg_{\wf}(\vect y_j)$, $j=1,\ldots,K$, contained in $\vect\Ac_{\wf}$ depend now on the unknown reflectivity 
vector $\bfrho$. This makes equations \eqref{eq:single} nonlinear with respect to $\bfrho$ and, hence, one would think that non-iterative inversion is impossible when multiple scattering is non-negligible. In fact, several nonlinear iterative methods have been proposed in the literature to solve this problem: see, for example, \cite{DMG05, MQ13}.
However, as demostrated below, by rearranging the terms in the equations, we can reformulate the problem to solve
for the locations of the scatterers directly (without any iteration), and then to recover their reflectivities in a second single step. 

To solve for the locations of the scatterers in one step, we introduce the {\em effective source vector} 
\begin{equation}
\label{eq:rhotilda}
\bfgamma_{\wf}=\diag(\bfrho)\mathbf{Z}^{-1}(\bfrho)\vect\Gc^T\vect\wf\, .
\end{equation}
Then,  using \eqref{responsematrix2}, \eqref{eq:single} can be rewritten as $\vect\Ac_{\wf}\bfrho=\vect\Gc\bfgamma_{\wf}=\vect b$, and the system
of equations
\begin{equation}\label{linearsystemsingleillum}
\vect\Gc\bfgamma_{\wf}=\vect b 
\end{equation}
becomes linear for the new unknowns $\bfgamma_{\wf}$. We point out that, unlike the problem considered in
\cite{CMP13}, when multiple scattering is not negligible, solving \eqref{linearsystemsingleillum} may not be
able to recover all the support of $\bfrho_0$. This is not a flaw of the formulation but an implicit
problem  of array imaging when multiple scattering is important. Indeed, due to multiple scattering effects 
it is possible that one or several
scalars $\wg_{\wf}(\vect y_j)$, $j=1,\ldots,K$, are very small or even zero and, hence, the corresponding scatterers 
become hidden. This is the well-known {\em screening effect} which makes scatterers undetectable, and
that it is manifested in our formulation making some of the components of the effective source vector
$\bfgamma_{\wf}$ 
arbitrary small. 

Note that, for a fixed
imaging configuration, the {\em screening effect} depends only on the illumination vector $\vect\wf$ and the amount of noise
in the data. Indeed, when the effective source at $\vect y_j$ is below the noise level because
$\wg_{\wf}(\vect y_j)$ is small, then the correponding scatterer cannot be detected.
This motivates us, in the next section, to consider active array imaging with multiple illuminations. In this case, active
array imaging is formulated as a joint sparsity recovery problem where we seek for an unknown matrix whose columns 
share the same support. By increasing
the diversity of illuminations, we minimize the {\em screening} and we have more chances of locating all the scatterers.

Since \eqref{linearsystemsingleillum} is underdetermined and the effective source vector $\bfgamma_{\wf}$
is sparse ($M\ll K$), we use $\ell_1$ minimization 
\begin{equation}\label{l1singleillum}
\min\|\bfgamma_{\wf}\|_{\ell_1}\quad\quad\text{s.t.}\quad\vect\Gc\bfgamma_{\wf}=\vect b
\end{equation}
to obtain $\bfgamma_{0\wf}$ from noiseless data.
When the data $\vect b$ is contaminated by a noise vector $\vect e$ with finite energy, we then seek the solution to the relaxed problem
\begin{equation}\label{l1singleillumnoise}
\min\|\bfgamma_{\wf}\|_{\ell_1}\quad\quad\text{s.t.}\quad\|\vect\Gc\bfgamma_{\wf}-\vect b\|_{\ell_2}<\delta \, ,
\end{equation}
for some given positive constant $\delta$.
Using Theorem\,$3.1$ in \cite{CMP13} and Theorem\,$14$ in \cite{TROPP06-1}, we have the following uniqueness and stability results.

\begin{thm}\label{thm.smv}
For a given array configuration, assume that the resolution of the IW is such that
\begin{equation}\label{mutualcoherence}
\max_{i\neq j}\left|\frac{\vect\wg_0^\ast(\vect y_i)\vect\wg_0(\vect y_j)}{\|\vect\wg_0(\vect y_i)\|_{\ell_2}\|\vect\wg_0(\vect y_j)\|_{\ell_2}}\right|<\epsilon,
\end{equation}
and there is no noise in the data.
If the number of scatterers $M$ satisfies that $M\epsilon<1/2$, then $\bfgamma_{0\wf}$ is the unique solution to \eqref{l1singleillum}. 
\end{thm}

\begin{thm}\label{thm.smvnoise}
Under the same condition \eqref{mutualcoherence} as in Theorem \ref{thm.smv}, if the data contain additive noise of finite energy $\|\vect e\|_{\ell_2}$, then
 the solution $\bfgamma_{\star\wf}$ to \eqref{l1singleillumnoise} satisfies
\begin{equation}\label{eq:SMVstability}
\|\bfgamma_{\star\wf}-\bfgamma_{0\wf}\|_{\ell_2}\le\frac{\delta}{\sqrt{1-(M-1)\epsilon}},
\end{equation}
provided $\delta\ge\|\vect e\|_{\ell_2}\sqrt{1+\frac{M(1-(M-1)\epsilon)}{(1-2M\epsilon+\epsilon)^2}}$.
Moreover, the support of $\bfgamma_{\star\wf}$ is fully contained in that of $\bfgamma_{0\wf}$, and all the components such that 
\begin{equation}\label{eq:SMVbound}
|(\bfgamma_{0\wf})_j|>\delta/\sqrt{1-(M-1)\epsilon}
\end{equation}
are within the support of $\bfgamma_{\star\wf}$.
\end{thm}
\begin{rem}
Theorem \ref{thm.smv} gives the required condition to recover the effective source vector
exactly from noiseless data. The resolution condition is based on the so called mutual coherence
\begin{equation}\label{def:mutualcoherence}
\mu(\vect\Gc) = \max_{i\neq j}\left|\frac{\vect\wg_0^\ast(\vect y_i)\vect\wg_0(\vect y_j)}{\|\vect\wg_0(\vect y_i)\|_{\ell_2}\|\vect\wg_0(\vect y_j)\|_{\ell_2}}\right|
\end{equation}
of the sensing matrix $\vect\Gc$, which is determined by the array imaging configuration (array size and resolution of the IW). The mutual coherence is a measure of how linearly independent the columns of the sensing matrix are. We give analytical results regarding the impact of the array geometry on \eqref{def:mutualcoherence} in
\S\ref{sec:arrayconfig}. Specifically, we show that a sensing matrix $\vect\Gc$ with small mutual coherence requires large arrays.
\end{rem}

Problems \eqref{l1singleillum} and \eqref{l1singleillumnoise} give the effective source vector $\bfgamma_{\wf}$. In a second step, we compute the true reflectivities from the solutions of these problems. According to \eqref{eq:rhotilda}, we need to solve a nonlinear
equation and, therefore, iteration seems to be inevitable. However, it is not necessary. 
Let $\Lambda_\star$ be the support of the recovered solution such that $|\Lambda_\star|=M'\leq M$, and $\bfgamma_{\wf,M'}$ the solution vector on that support. 
From \eqref{eq:gFL} and \eqref{eq:rhotilda}, we obtain
$$
\bfgamma_{\wf, M'}=\diag(\mathbf{Z}^{-1}(\bfrho_{M'})\vect\Gc^T\vect\wf)\bfrho_{M'} = \diag(\wg_{\wf}(\vect y_{n_1}),\ldots,\wg_{\wf}(\vect y_{n_{M'}}))\bfrho_{M'}\, ,
$$
where the scalars $\wg_{\wf}(\vect y_{n_j})=\vect\wg^T_{FL}(\vect y_{n_j})\vect\wf$.
Note that the scalars $\wg_{\wf}(\vect y_{n_j})$ are the exciting fields at the scatterer's positions, that is,
$\wg_{\wf}(\vect y_{n_j}):=\widehat\psi_j^e(\vect y_{n_1},\ldots,\vect y_{n_{M'}})$, and that the effective sources $\gamma_{n_j}$
are the true reflectivities $\rho_{n_j}$ of the scatterers multiplied by the exciting fields. Hence, using \eqref{eq:effectivefields}, we can compute $\wg_{\wf}(\vect y_{n_j})$ explicitly as follows
\begin{equation}\label{invertless}
\wg_{\wf}(\vect y_{n_j})=\vect\wg_0^T(\vect y_{n_j})\vect\wf+\sum_{k=1,k\neq j}^{M'}\gamma_{k}\wG_0(\vect y_{n_j},\vect y_{n_k}),\quad j=1,\ldots,{M'}.
\end{equation}
Then, the true reflectivities of the scatterers are given by
\begin{equation}\label{eq:truereflectivities}
\rho_{n_j}=\gamma_{n_j}/\wg_{\wf}(\vect y_{n_j}),\quad j=1,\ldots,{M'}.
\end{equation}
For the noiseless case, $\Lambda_\star=\Lambda_0$ based on Theorem \ref{thm.smv}. When the data contains additive noise, we choose the support $\Lambda_\star$ 
of the solution recovered by \eqref{l1singleillumnoise} such that all the components of $\bfgamma_{\wf,M'}$
satisfy \eqref{eq:SMVbound}.

To summarize, when a single illumination is used to probe the medium, we take two steps 
to locate the scatterers and to obtain their reflectivities, as follows.
\begin{itemize}
\item Solve the $\ell_1$ minimization problem \eqref{l1singleillum} or \eqref{l1singleillumnoise} for the {\it effective source vector}.
\item Compute the true reflectivities using \eqref{invertless} and \eqref{eq:truereflectivities} on the support $\Lambda_\star$.
\end{itemize}
There are many fast and efficient numerical algorithms for solving \eqref{l1singleillum} or \eqref{l1singleillumnoise}. In the simulation study below, we use the 
iterative shrinkage-thresholding algorithm GelMa, described in \cite{MNPR12}, due to its flexibility with
respect to the choice of the regularization parameter used in the algorithm.

\section{Imaging using multiple illumination vectors}\label{sec:MMV}

In the previous section we discuss a non-iterative approach for array imaging with multiple scattering when a single
illumination is used. Although the proposed approach can recover the locations and reflectivities of the scatterers exactly when the data is noiseless, 
it can be very sensitive to additive noise, especially when the noise level is high, leading to unreliable images. Moreover, the
{\em screening effect} associated with multiple scattering can cause the failure of recovering some scatterers in the IW.
This means that for a given illumination $\vect\wf$ the number of effective sources $M'$ is strictly less than the number 
of scatterers $M$.
These two problems can be mitigated by using multiple illuminations which can often be controlled to increase
the power of the signals received at the array.
We will show that by carefully  choosing the illumination vectors, the use of multiple inputs makes array imaging
more stable in the presence of relatively high noise and, at the same time, the screening effect is minimized.

\subsection{Imaging with multiple arbitrary illuminations}

To work with data generated by multiple (random) illumination vectors, a natural extension is to stack the
data vectors $\vect b^j$ from illuminations $\vect\wf^j$, $j=1,\ldots,\nu$, into a single $\nu N$ vector, and 
to apply the approach in \S\ref{sec:SMV} to the augmented linear system. 
However, by simply stacking the data forming a larger linear system  not only increases the dimensionality of the problem
but also fails to exploit the intrinsic relation among the multiple data vectors. To make use of the data structure, we formulate
the problem of array imaging with multiple illuminations as a joint sparsity recovery problem, also known as the
{\it Multiple Measurement Vector} (MMV) approach. Instead of solving a matrix-vector equation for the unknown reflectivity vector, 
we now solve a matrix-matrix equation for an unknown
matrix variable whose columns share the same sparse support but possibly different nonzero values. The MMV approach has been
widely studied in passive source localization problems and other applications with success, see for example \cite{MCW05}.
With the introduction of the {\it effective source vector}, MMV can also be used effectively for active array imaging when multiple scattering between scatterers is important.

Let $\mathbf{B}=[\vect b^1\,\ldots\,\vect b^\nu]$ be the matrix whose columns are the data vectors generated by all the illuminations, and $\vect{\Rho}=[{\bfgamma}^1\,\ldots\,{\bfgamma}^\nu]$ be the unknown matrix whose
$j^\mathrm{th}$ column corresponds to the {\em effective source vector}
$\bfgamma^j$ under illumination $\vect\wf^j$, $j=1,\ldots,\nu$. Then, the MMV
formulation for active array imaging is to solve for $\vect\Rho$ from the matrix-matrix equation
\begin{equation}\label{linearsystemmmv}
\vect\Gc \vect\Rho = \mathbf{B}.
\end{equation}
In this framework, the sparsity of the matrix variable $\vect\Rho$ is characterized by the number of nonzero
rows of the matrix. 
More precisely, we define the row-support of
a given matrix $\vect\Rho$ by
\[\operatorname{rowsupp}(\vect\Rho)=\{i:\,\,\exists \, j\,\,\text{s.t.}\,\, \Rho_{ij}\neq0\}\,,\]
which is equivalent to
\[\operatorname{rowsupp}(\vect{\Rho})=\{i:\,\,\|X_{i\cdot}\|_{\ell_p}\neq0\},\]
where $p\ge1$. From this definition, we see that when the matrix $\vect\Rho$ degenerates to a
column vector, the row-support reduces to the support of the vector.
The joint sparsity of $\vect\Rho$ is then measured
by the row-wise $\ell_0$ norm of $\vect\Rho$ defined by
\[\Xi_0(\vect\Rho)=|\operatorname{rowsupp}(\vect\Rho)|.\]
With these definitions, the sparsest solution of array imaging using multiple illuminations is given by
the solution to the problem
\begin{equation}\label{MMV.NP}
\min\Xi_0(\vect\Rho)\quad\text{s.t.}\quad\vect\Gc\vect\Rho =\mathbf{B}.
\end{equation}
Similarly to the $\ell_0$ norm minimization problem in SMV, \eqref{MMV.NP} is an NP hard
problem. An alternative is to solve the convex relaxed problem
\begin{equation}\label{MMV.convex}
\min\Xi_1(\vect\Rho)\quad\text{s.t.}\quad\vect\Gc\vect\Rho=\mathbf{B},
\end{equation}
where the substitution of $\Xi_0$ by a certain function $\Xi_1$ turns \eqref{MMV.NP} into a tractable problem. There are many choices of $\Xi_1$ as discussed, for example, in \cite{Cotter05,CH06,TROPP06-2}.
We note here that  $\Xi_1=J_{p,1}$ for any $p\ge1$, as defined in \eqref{eq:Jpq},
can be used to replace the nonconvex objective function $\Xi_0$.
We will use $p=2$ in the following discussion which has been studied in, 
for example, \cite{Cotter05,MCW05,CH06,ER10}.
Therefore, we consider the following convex relaxed problem to image 
the scatterers with multiple illumination vectors
\begin{equation}\label{MMV21}
\min J_{2,1}(\vect\Rho)\quad\text{s.t.}\quad\vect\Gc\vect\Rho=\mathbf{B}.
\end{equation}
Similar to Theorem~\ref{thm.smv}, we have the following condition for recovery using \eqref{MMV21}.
\begin{thm}\label{thm.mmv}
For a given array configuration, assume that the resolution of the IW satisfies \eqref{mutualcoherence}.
If the number of scatterers $M$ is such that $M\epsilon<1/2$, then
$\vect\Rho_0=[\tilde{\bfrho}^1_0\,\ldots\,\tilde{\bfrho}^\nu_0]$
is the unique solution to \eqref{MMV21}.
\end{thm}
\begin{rem}\label{rem:mmv}
The condition given in Theorem~\ref{thm.mmv} is also the sufficient condition for the complete family of MMV problems that use the
$J_{p,1}$ type of objective function to convert the original non-convex problem \eqref{MMV.NP} into a convex, solvable one. In fact,
we prove Theorem~\ref{thm.mmv} by showing $\vect\Rho_0$ is the unique solution to
\[\min J_{p,1}(\vect\Rho)\quad\text{s.t.}\quad\vect\Gc\vect\Rho=\mathbf{B}\]
for any $1<p<\infty$ in Appendix~\ref{proof0}. The case of $p=\infty$ is studied in \cite{TROPP06-2}.
We also note that for the case $p=1$, the resulting formulation becomes fully decoupled. Indeed, solving
\[\min J_{1,1}(\vect\Rho)\quad\text{s.t.}\quad\vect\Gc\vect\Rho=\mathbf{B}\]
can be viewed as solving $\nu$ simple $\ell_1$-norm minimization problems with single illumination, and hence,
this approach does not fully utilize the joint sparsity of the problem. Therefore, the support is not
simutaneously recovered with $J_{1,1}$. This observation has also been discussed in \cite{TROPP06-2} and \cite{CH06}.
\end{rem}
When the collected data is contaminated by additive noise vectors $\vect{e}^j$, $j=1,\ldots,\nu$, equations \eqref{linearsystemmmv} become
\begin{equation}\label{linearsystemmmvnoise}
\vect\Gc\vect\Rho=\mathbf{B}+\vect{\Ec}\, .
\end{equation}
Here, $\vect{\Ec}=[\vect{e}^1\cdots\vect{e}^\nu]$
is the matrix whose columns are independent noise vectors $\vect{e}^j$
corresponding to each illumination
vector $\vect\wf^j$, $j=1,\ldots,\nu$.
Then, similar to the the single illumination case, we seek a solution to
\begin{equation}\label{MMV21noise}
\min J_{2,1}(\vect\Rho)\quad\text{s.t.}\quad\|\vect\Gc\vect\Rho-\mathbf{B}\|_F<\delta \, ,
\end{equation}
for some pre-specified constant $\delta$. As stated in the following result,
the solution to \eqref{MMV21noise} recovers the sparsest solution $\vect\Rho_0$
upon certain error bound. The result is proved using a similar approach as the one used in \cite{TROPP06-2} for $J_{\infty,1}$.
Details are given in Appendix~\ref{proof}.
\begin{thm}\label{thm.mmvnoise}
For a given array configuration, assume that the resolution of the IW satisfies \eqref{mutualcoherence}.
If the number of scatterers $M$ is such that $M\epsilon<1/2$, and
\begin{equation}\label{eq:constraintcondition}
\delta\ge\|\vect{\Ec}\|_F\sqrt{1+\frac{M(1-(M-1)\epsilon)}{(1-2M\epsilon+\epsilon)^2}},
\end{equation}
then \eqref{MMV21noise} has a unique solution $\vect\Rho_\star$
which has row support included in that of $\vect\Rho_0$ and satisfies
\begin{equation}\label{eq:MMVstability}
\|\vect\Rho_\star-\vect\Rho_0\|_F\le\frac{\delta}{\sqrt{1-(M-1)\epsilon}}.
\end{equation}
Moreover, the row support of $\vect\Rho_\star$ contains all the rows $i$ satisfying 
\begin{equation}\label{eq:MMVbound}
\|(\vect\Rho_0)_{i\cdot}\|_{\ell_2}>\frac{\delta}{\sqrt{1-(M-1)\epsilon}}.
\end{equation}
\end{thm}
According to Theorems~\ref{thm.mmv} and \ref{thm.mmvnoise} the performance of \eqref{MMV21} and
\eqref{MMV21noise} does not depend on the number of measurements $\nu$. Therefore, judging from these
theoretical results, there is no quantitative  improvement in the conditions imposed on the imaging setup 
when using multiple illuminations compared to those for a single illumination.
Intuitively, this is so because it is possible that measurements from different (random) illuminations
may all be rather ineffective and, therefore, there would not be an advantage in using multiple measurements in such a case.
However, in practice, we observe that there is in general improvement in the image, which is much better when (random)
multiple illuminations are used, especially in the presence of additive noise. To explain the improved performance seen in practice,
the authors in \cite{ER10} carried out
an average-case analysis of the underlying joint sparsity recovery problem by introducing a probability model for $\vect\Rho$.
They showed in that context that the probability of failing to recover the 
true solution vector decays exponentially with the number of measurements.

We note that the recovery condition of \eqref{MMV21} and \eqref{MMV21noise} still depends
on the mutual coherence of the sensing matrix $\vect\Gc$, i.e., on \eqref{mutualcoherence}.
As we have already remarked, this condition depends only on the configuration of the imaging problem, the array geometry and
the chosen discretization of the image window IW. In \S\ref{sec:arrayconfig}, we discuss array configurations that lead to
different conditions \eqref{mutualcoherence}.

Once we obtain from \eqref{MMV21} or \eqref{MMV21noise} the matrix $\vect\Rho_\star$, whose columns are the effective sources
corresponding to the different illuminations, we then compute in a second step the true reflectivities as follows.
For each component $i$ in the support such that \eqref{eq:MMVbound} is satisfied,
we compute the reflectivities $\rho^j_{i}$ corresponding to each illumination $j$ by applying
\eqref{invertless} and \eqref{eq:truereflectivities}.
We then take the average $\frac{1}{\nu}\sum_{j=1}^\nu\rho_{i}^j$ as the estimated reflectivity.

\subsection{Imaging with optimal illuminations}
In order to increase the robustness of the methods \eqref{l1singleillum} and \eqref{l1singleillumnoise},
and to mitigate screening effects, MMV uses data obtained from multiple illuminations. 
One approach in MMV is to use multiple illuminations selected randomly. However, such illuminations
may not avoid screening above certain noise level, as we see in numerical simulations in \S\ref{simulation}.
Furthermore, using random illuminations may not be very efficient because a large number of them are needed
to get a significant improvement in the image.

We now introduce an approach that uses optimal illuminations within the MMV framework.
The use of optimal illuminations for array imaging in homogeneous and random media has been studied in \cite{BPT06,BPT07,CMP13}.
The optimal illuminations can be computed systematically from the singular value decomposition (SVD) of the array
response matrix $\vect\wP$, or with an iterative time reversal process as discussed in \cite{PTF95,MTF04} when the full
array response matrix is not available.
Let the SVD of $\vect\wP$ given in \eqref{responsematrix1} be
\[\vect\wP=\vect\wU\vect\Sigma\vect\wV^\ast=\sum_{j=1}^{\tilde M}\sigma_j\wU_{\cdot j}\wV_{\cdot j}^\ast,\]
where $\wU_{\cdot j}$ and $\wV_{\cdot j}$ are the left and right singular vectors, respectively, and the nonzero singular values
$\sigma_j$ are given in descending order as $\sigma_1\ge\sigma_2\ge\cdots\ge\sigma_{\tilde M}>0$, with $\tilde M\ge M$.
When there is no additive noise in the data, we have $\tilde M=M$. Let the illumination vectors be the right singular vectors
$\wV_{\cdot j}$, that is, $\vect\wf^j=\wV_{\cdot j}$, $j=1,\ldots,\nu\le\tilde M$. Then,
\begin{equation}\label{linearsystemmmv optimal}
\mathbf{B}_{opt}=\vect\Gc\vect\Rho=\vect\wP\vect\wV_{\cdot,1:\nu}=[\sigma_1\wU_{\cdot 1}\cdots\sigma_\nu\wU_{\cdot\nu}]+\vect{\widetilde{\Ec}}.
\end{equation}
All the information for imaging is contained in the matrix $\mathbf{B}_{opt}$ given in \eqref{linearsystemmmv optimal}.
It is also clear that
the use of optimal illuminations filters out noise in the data because it reduces the dimensionality of the resulting optimization
problem without loss of essential information about the scatterers.

Recall that the singular vectors $\wV_{\cdot j}$, with $j=1,\ldots,M$, are the illuminations that focus at each scatterer 
when multiple scattering is negligible and the scatterers are well resolved by the array. The key point here is that when
multiple scattering is important, these optimal illuminations still deliver most of the energy around the scatterers,
but each $\wV_{\cdot j}$ is no longer associated with a single scatterer only. All the scatterers are illuminated in general
where multiple scattering is important. As a consequence, taking a few
top singular vectors, less than $\tilde M$, is enough to locate all the scatterers and image them.
Moreover, taking fewer illuminations can be beneficial since less noise is introduced into \eqref{linearsystemmmv optimal}. 
We illustrate this observation with numerical examples in \S\ref{simulation}.

We note that, by using optimal illuminations from the SVD of the array response matrix $\vect\wP$,
we are able to make the
performance of the MMV formulation deviate significantly from the average case when using random illuminations.

\subsection{A sparsity promoting algorithm}
The MMV problem \eqref{MMV.NP} can be solved 
by greedy algorithms that are straightforward generalizations of orthogonal matching pursuit for the single measurement case \cite{Cotter05, Duarte05,TROPP06-2-0,Gribonval08}.  At each iteration, these algorithms increase  the joint support set of the estimated solution by one index, until a given number of  columns vectors of the sensing matrix are selected or the approximation error is below a preset threshold.
Sparse Bayesian learning approaches developed for the single measurement case have also been extended to solve \eqref{MMV.NP} \cite{Wipf07, Zhang11}. Both types of methods, however, become slow when the size of the problem is large.
Alternatively, \eqref{MMV.NP}  can be relaxed to the convex formulation \eqref{MMV21} (or \eqref{MMV21noise}) and then consider algorithms that are extensions of those used to solve \eqref{l1singleillum} (or  \eqref{l1singleillumnoise}).

For our numerical simulations we will employ an extension of an iterative algorithm proposed in \cite{MNPR12}, called GeLMA. This
is a shrinkage-thresholding algorithm for solving $\ell_1$-minimization problems which has proven to be very efficient and whose
solution does not depend on the regularization parameter that promotes sparse solutions.
In our case, the algorithm deals with the penalized problem
\begin{equation}\label{MMV-functional}
L(\vect\Rho)=\frac{1}{2}\|\vect\Gc\vect\Rho-\mathbf{B}\|_F^2+\tau J_{2,1}(\vect\Rho)\, ,
\end{equation}
and is derived based on the augmented Lagrangian form
\begin{equation}\label{MMV-functional-augmented}
F(\vect\Rho,\vect\Zc)=  L(\vect\Rho) + \ltr \vect\Zc , \mathbf{B} - \vect\Gc\vect\Rho \rtr.
\end{equation}
For any fixed matrix multiplier $\vect\Zc$, the function $F(\vect\Rho,\vect\Zc)$ is convex in $\vect\Rho$ and thus, we can compute
its minimum iteratively. At iteration $(k+1)$, we first fixed $\vect\Zc = \vect\Zc^{(k)}$ and we seek the
minimum of the differentiable part of $F(\vect\Rho,\vect\Zc^{(k)})$ as
\[\mathbf{Y}^{(k+1)} = \argmin_{\vect\Rho}\left\{\ \frac{1}{2}\|\vect\Gc\vect\Rho-\mathbf{B}\|_F^2 + \ltr \vect\Zc^{(k)} , \mathbf{B} - \vect\Gc\vect\Rho \rtr\right\} .\]
Together with $\vect\Rho^{(k)}$ from the previous iteration, we compute 
\[\mathbf{Y}^{(k+1)} =\vect\Rho^{(k)} + \beta\vect\Gc^\ast(\vect\Zc^{(k)}+\mathbf{B} - \vect\Gc\vect\Rho^{(k)})\, 
\]
using a first order iterative gradient descent method, where $\beta$ is the step size.
Next, we consider the (non-differentiable) regularization part through minimizing
\[\min_{\vect\Rho}\left\{\frac{1}{2}\|\vect\Rho-\mathbf{Y}^{(k+1)}\|_F^2+\beta\tau J_{2,1}(\vect\Rho)\right\}.\]
Due to the row separability of both, the Frobenius matrix norm and the function $J_{2,1}$, 
this problem can be decomposed into the following $N$ sub-problems
\[
\min_{\Rho_{i\cdot}}\left\{\frac{1}{2}\|\Rho_{i\cdot}-Y_{i\cdot}^{(k+1)}\|_{\ell_2}^2+\beta\tau\|\Rho_{i\cdot}\|_{\ell_2}\right\},\quad i=1,\ldots,N.
\]
Each sub-problem is quadratic in $\Rho_{i\cdot}$, and there exists a closed-form solution given by
\[\Rho_{i\cdot}^{(k+1)}=\operatorname{sign}(\|Y_{i\cdot}^{(k+1)}\|_{\ell_2}-\beta\tau)\frac{\|Y_{i\cdot}^{(k+1)}\|_{\ell_2}-\beta\tau}{\|Y_{i\cdot}^{(k+1)}\|_{\ell_2}}Y_{i\cdot}^{(k+1)},\quad i=1,\ldots,N \, ,\]
which involves only a simple shrinkage-thresholding step.
Finally, $\vect\Zc^{(k+1)}$ is found by applying a gradient ascent method as
\[\vect\Zc^{(k+1)}=\vect\Zc^{(k)} + \beta \, (\mathbf{B}- \vect\Gc\vect\Rho^{(k)}) .\]
For more details regarding the properties of this algorithm for the single measurement case, we refer to \cite{MNPR12} and references therein. We summarize it  for MMV problems in Algorithm~\ref{algo}.
\begin{algorithm}
\begin{algorithmic}
\REQUIRE Set $\vect\Rho=\vect0$, $\vect\Zc=\vect 0$ and pick the step size $\beta$, and the regularization parameter $\tau$
\REPEAT
\STATE Compute the residual $\vect\Rc= \mathbf{B} - \vect\Gc\vect\Rho$ 
\STATE $\vect\Rho\Leftarrow\vect\Rho + \beta\vect\Gc^\ast(\vect\Zc + \vect\Rc)$
\STATE $\Rho_{i\cdot}\Leftarrow\operatorname{sign}(\|\Rho_{i\cdot}\|_{\ell_2}-\beta\tau)\frac{\|\Rho_{i\cdot}\|_{\ell_2}-\beta\tau}{\|\Rho_{i\cdot}\|_{\ell_2}}\Rho_{i\cdot}$, $i=1,\ldots,K$
\STATE $\vect\Zc=\vect\Zc + \beta\vect\Rc$
\UNTIL{Convergence}
\end{algorithmic}
\caption{GelMa-MMV for solving \eqref{MMV-functional-augmented}}
\label{algo}
\end{algorithm}

\section{Array configuration and mutual coherence}\label{sec:arrayconfig}
We have already discussed that the performance of sparsity promoting algorithms strongly
depends on the mutual coherence of the sensing matrix, which is related to the array imaging
configuration. In this section, we give some analytical results for the mutual coherence
of two types of arrays that are often used in array imaging: planar arrays and
spherical arrays. The schemata of these two types are illustrated in Figure~\ref{fig:arrayschema}.
We show that under similar configurations of the IW (distance
to the array and the resolution), spherical arrays give smaller upper bounds of the
inner products of the normalized Green's function vectors than planar arrays in
condition \eqref{mutualcoherence}.
We give the proofs in Appendix~\ref{appendix:innerproduct}.

\begin{figure}
\begin{center}
\begin{tabular}{cc}
\includegraphics[scale=0.4]{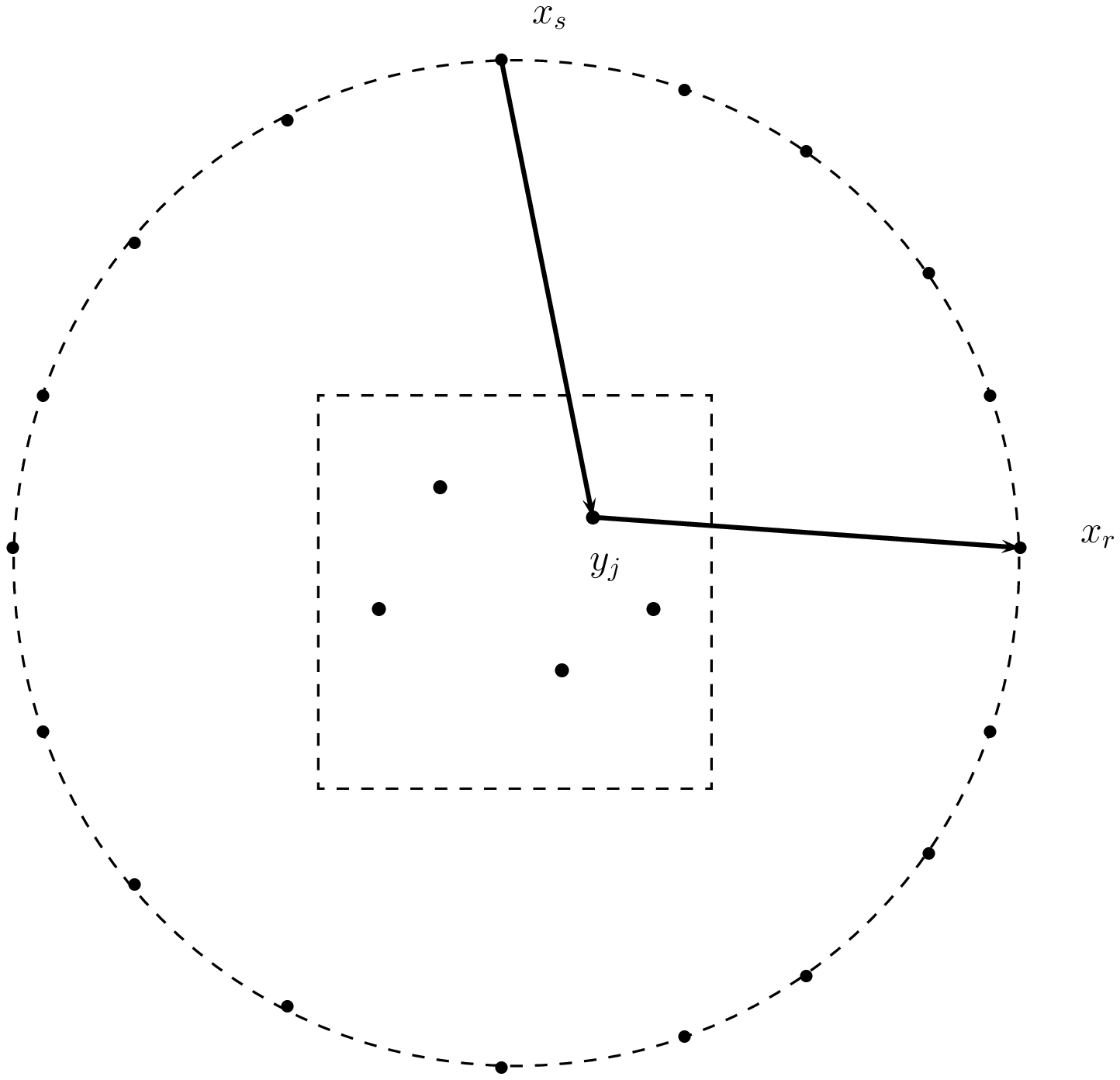}&
\includegraphics[scale=0.5]{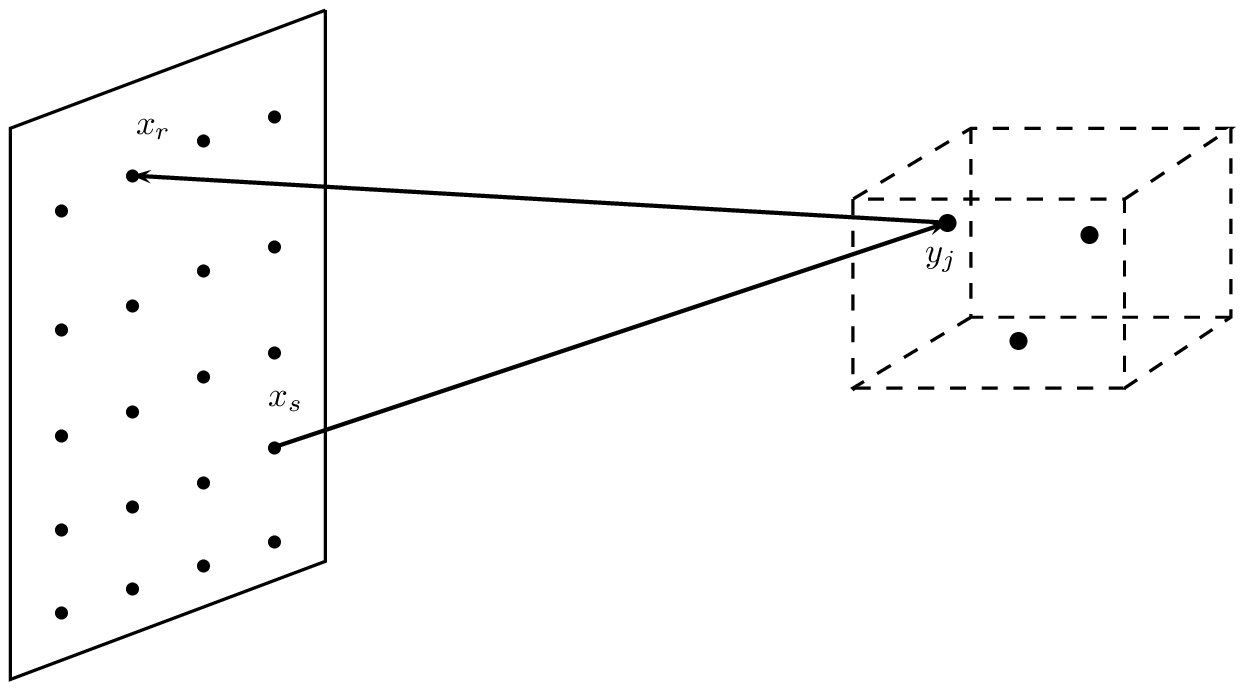}
\end{tabular}
\caption{Schemata of spherical (left) and planar (right) arrays}
\label{fig:arrayschema}
\end{center}
\end{figure}

The first result is on the estimate of the inner product when a spherical array surrounding the IW is used. It is
a well-known classical result. We state it here to be self-contained.
\begin{prop}\label{prop:spherical array}
Assume that the IW is fully surrounded by a spherical array of radius $L$.
Given any two points $\vect y_k$ and
$\vect y_{k'}$ in the IW such that $\lambda\ll|\vect y_k-\vect y_{k'}|\ll L$, we have 
\begin{equation}\label{eq:sphericalarrayinnerproduct}
\frac{\vect\wg_0^\ast(\vect y_k)\vect\wg_0(\vect y_{k'})}{\|\vect\wg_0(\vect y_k)\|_{\ell_2}\|\vect\wg_0(\vect y_{k'})\|_{\ell_2}}\approx\sinc(\kappa|\vect y_k-\vect y_{k'}|)\sim\frac{1}{\kappa|\vect y_k-\vect y_{k'}|}\, .
\end{equation}
Hence, the mutual coherence condition of the sensing matrix for spherical arrays is improved
at the rate of the pixel size relative to the wavelength.
\end{prop}

The second result is on the estimate when a planar array is used.
\begin{prop}\label{prop:planar array}
Assume a planar array of finite size and let $\vect y_k$ and
$\vect y_{k'}$ be two points within the IW such that $\lambda\ll|\vect y_k-\vect y_{k'}|\ll L$. Then, we have
\begin{equation}\label{eq:planararrayinnerproduct}
\frac{\vect\wg_0^\ast(\vect y_k)\vect\wg_0(\vect y_{k'})}{\|\vect\wg_0(\vect y_k)\|_{\ell_2}\|\vect\wg_0(\vect y_{k'})\|_{\ell_2}}\sim\frac{1}{\sqrt{\kappa|\vect y_k-\vect y_{k'}|}}.
\end{equation}
Hence, the mutual coherence condition of the sensing matrix for planar arrays is improved
at the rate of square root of the pixel size relative to the wavelength.
\end{prop}

Based on these results, the upper bound of \eqref{def:mutualcoherence} is smaller for spherical arrays than for planar arrays.
The pixel size of the IW with which good images are obtained is smaller for spherical arrays than for planar arrays.
According to the analyses in \S\ref{sec:SMV} and \S\ref{sec:MMV}, array imaging with
spherical arrays can then locate more scatterers with higher resolution and is more robust with
respect to the additive noise than array imaging with planar arrays, provided all other conditions
are identical. This observation is supported by the numerical experiments.

\section{Numerical simulation}\label{simulation}
In this section we present numerical simulations  in two dimensions. The linear array 
consists of $100$ transducers that are one wavelength $\lambda$ apart. Five
scatterers are placed within an IW  of size $41\lambda\times41\lambda$ which is at a distance $L=100\lambda$ from the linear array. 
The amplitudes of the reflectivities of the scatterers, $|\alpha_j|$, are $2.96$, $2.76$, $2.05$, $1.54$ and $1.35$ (see
Fig.~\ref{original}).  Their phases are set randomly in each realization. We note that, 
given an illumination vector $\vect\wf$ and a scatterer configuration $\vect\rho_0$ with fixed amplitudes,
the exact amount of multiple scattering depends on the realization
of the phases in $\vect\rho_0$. For the amplitudes of the reflectivities chosen here, the amount of multiple scattering, defined by
 \begin{equation}\label{ams}
\frac{\|\vect\wP - \vect\wP_{ss}\|_F}{\|\vect\wP_{ss}\|_F}\times 100\, ,
\end{equation}
typically ranges between $50\%$  and $100\%$ in the simulations shown below. In \eqref{ams}, $\vect\wP_{ss}$ is
the response matrix without multiple scatterering, computed by replacing
$\vect\Gc_{FL}^T(\bfrho_0)$ by $\vect\Gc^T$ in \eqref{responsematrix3}, i.e., $\vect\wP_{ss}=\vect\Gc\diag(\bfrho)\vect\Gc^T$.

The five scatterers are within an IW that is discretized using a uniform lattice with points separated by one wavelength $\lambda$.
This results in a $41\times41$ uniform mesh. Hence, we have $1681$ unknowns and $100$ measurements.
In all the images shown below, we normalize the spatial units by the wavelength $\lambda$.
For this configuration of the IW, the mutual coherence \eqref{def:mutualcoherence} of the sensing matrix $\vect\Gc$ has a numerical
value equal to $0.98$. This, together with $M=5$ scatterers,
clearly violates the sufficient condition for stable reconstruction required by formulations using either single illumination
or multiple illuminations.
However, this condition is quite conservative and we will show that the images are still good when the noise level is low in the data. 
Finally, we note that the obtained images depend on the realization of the random phases of the scatterers.
In all the images shown below, we do not display the ones with the best quality we have seen in our numerical study.

\begin{figure}[htbp]
\begin{center}
\includegraphics[scale=0.35]{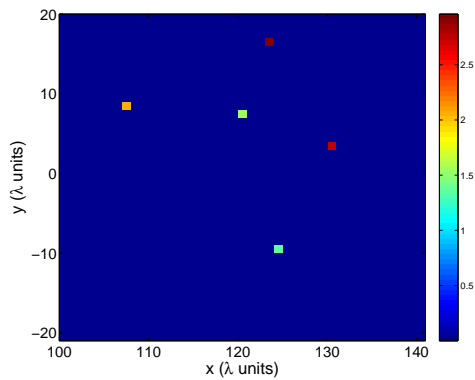}
\caption{Original configuration of the scatterers in a $41\times41$ image window with grid points separated by $1$.
The amplitudes of the reflectivities of the scatterers, $|\alpha_j|$, are $2.96$, $2.76$, $2.05$, $1.54$ and $1.35$.}
\label{original}
\end{center}
\end{figure}

Figure~\ref{SVM} shows the results of  $\ell_1$ norm minimization with $0\%$ (left), $10\%$ (middle) and $20\%$ noise (right)
when a single illumination coming from the center of the array is used. The exact locations of the scatterers in these images
are indicated with small white dots. When there is no noise in the data, $\ell_1$ norm minimization recovers the positions and
reflectivities of the scatterers accurately. However, when $10\%$ and $20\%$ of noise is added to the data, the method fails to
recover some of the scatterers and the images show some ghosts. Note that some scatterers are missing in the middle and right images of Figure~\ref{SVM}.

\begin{figure}
\centering
\begin{tabular}{ccc}
\includegraphics[scale=0.25]{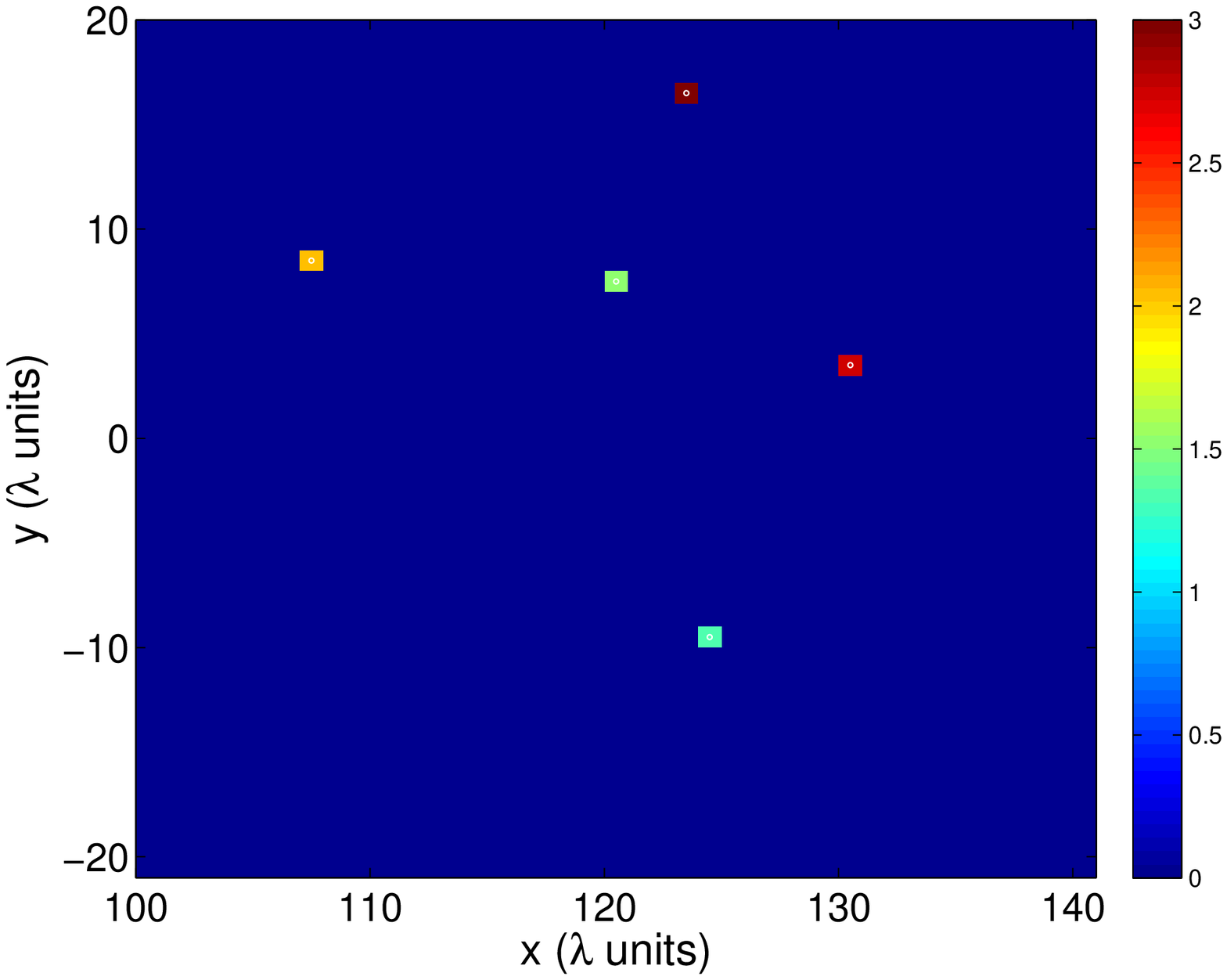} & 
\includegraphics[scale=0.25]{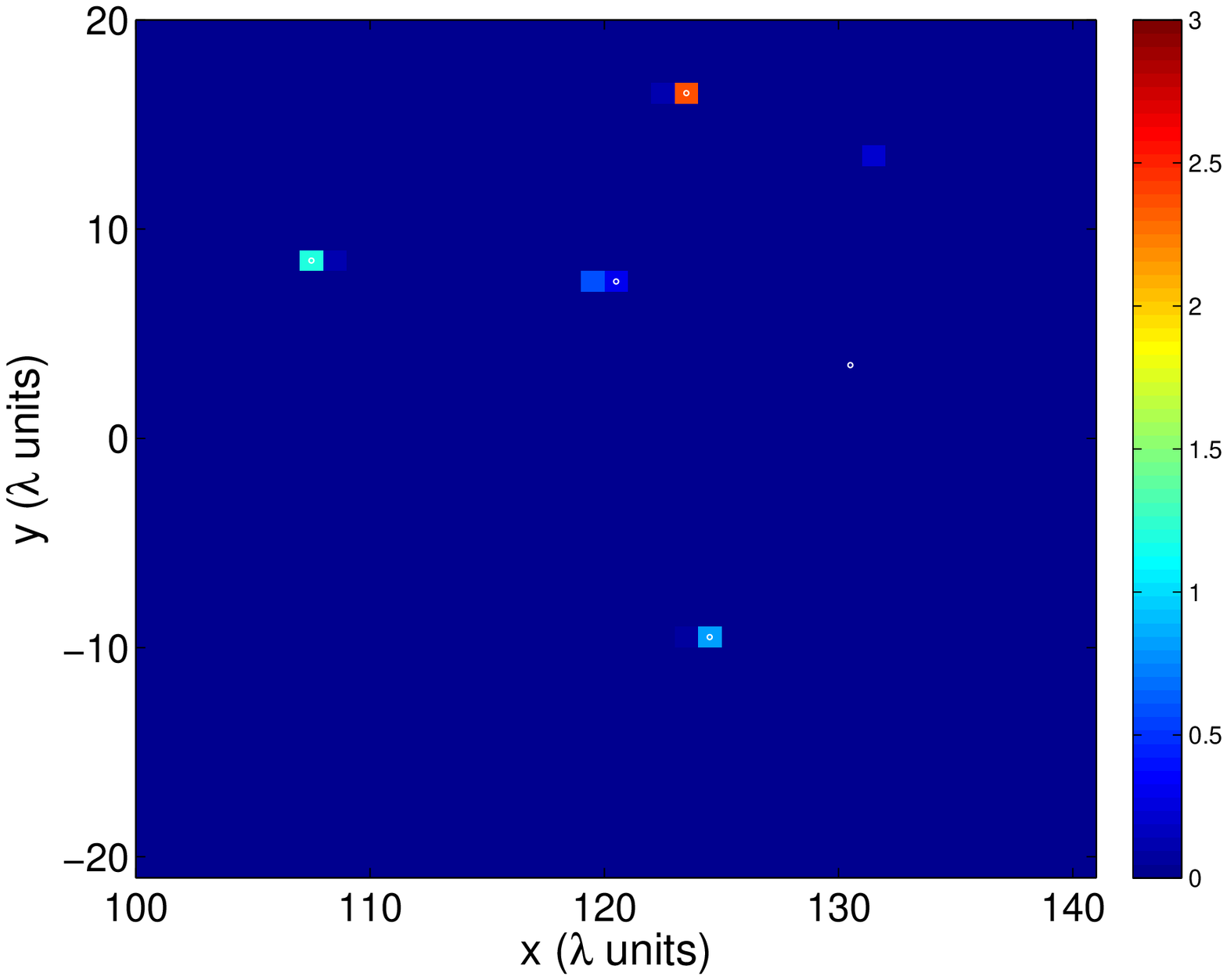} & 
\includegraphics[scale=0.25]{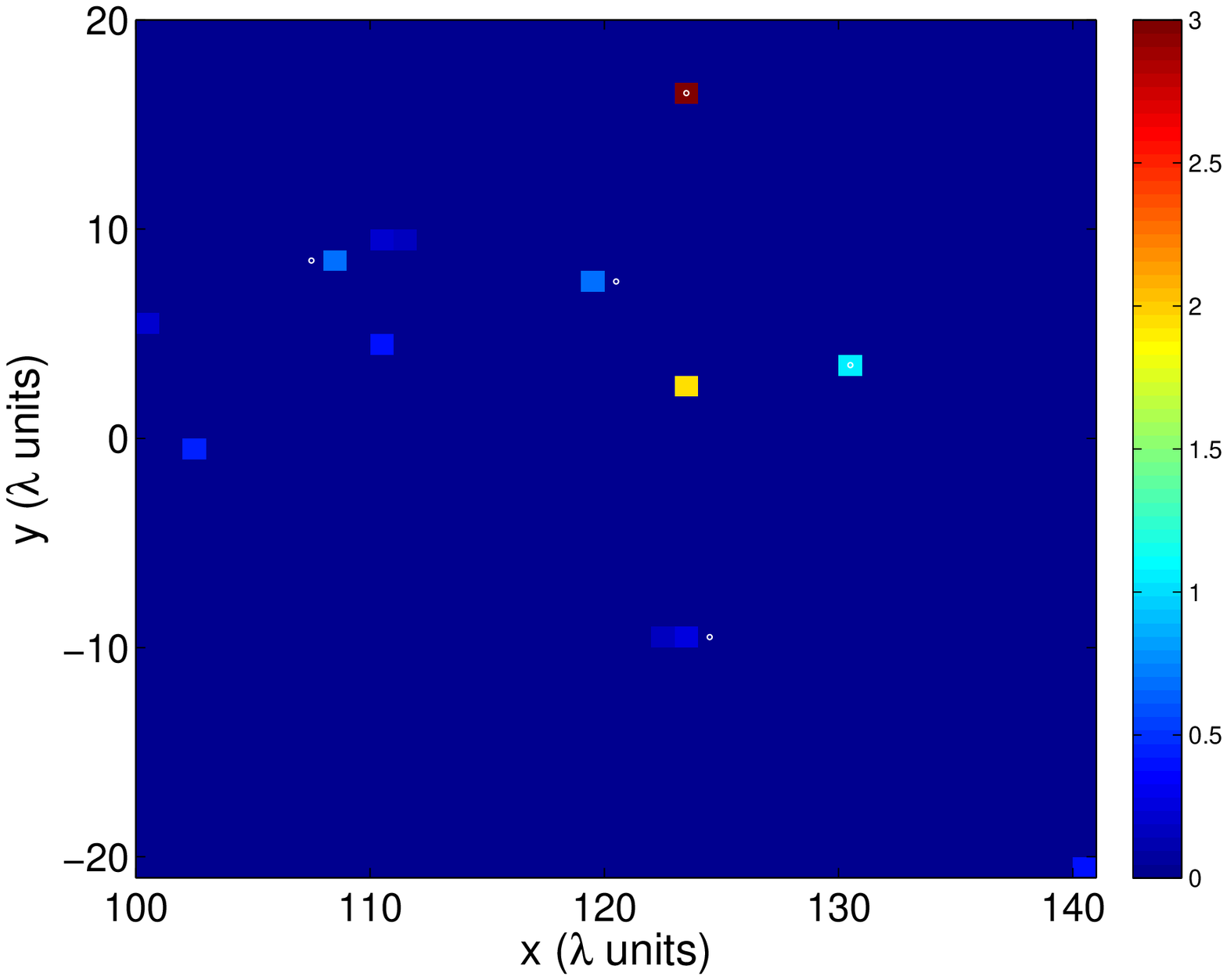}
\end{tabular}
\caption{Images reconstructed by solving \eqref{l1singleillum} and \eqref{l1singleillumnoise} when single illumination is used.
From left to right, there is $0\%$, $10\%$, and $20\%$ noise in the data.
}
\label{SVM}
\end{figure}

In order to stabilize the images, we study the improvement of the results when data collected with multiple illuminations are used.
We consider first the case where the illuminations are randomly selected. By random illuminations we mean several illuminations
coming, each one, from only one of the transducers on the array at a time, i.e., $\wf_{p}=1$ and  $\wf_q=0$ for $q\neq p$,
with $p$ chosen randomly at a time.
Figure~\ref{MMV} shows the results of the MMV algorithm when $5$ (top row) and $15$ (bottom row)
random illuminations are used. Additive noise at level $10\%$ (left column), $20\%$ (middle column) and $50\%$ (right column) is
added to the data in these numerical experiments. As expected, the images obtained with multiple illuminations are more stable with
respect to additive noise. In fact, only a small number of illuminations are needed to improve the imaging performance significantly.
However, it is not always true in general that more random illuminations always lead to better images. For instance, the image
obtained with $20\%$ noise and $15$ random illuminations (middle image of the bottom row) is worse than that obtained with $20\%$
noise and $5$ random illuminations (middle image of the top row). This is so because the illuminations are chosen randomly
and ``good" illuminations that lead to enough data diversity cannot be guaranteed. 

\begin{figure}
\centering
\begin{tabular}{ccc}
\includegraphics[scale=0.25]{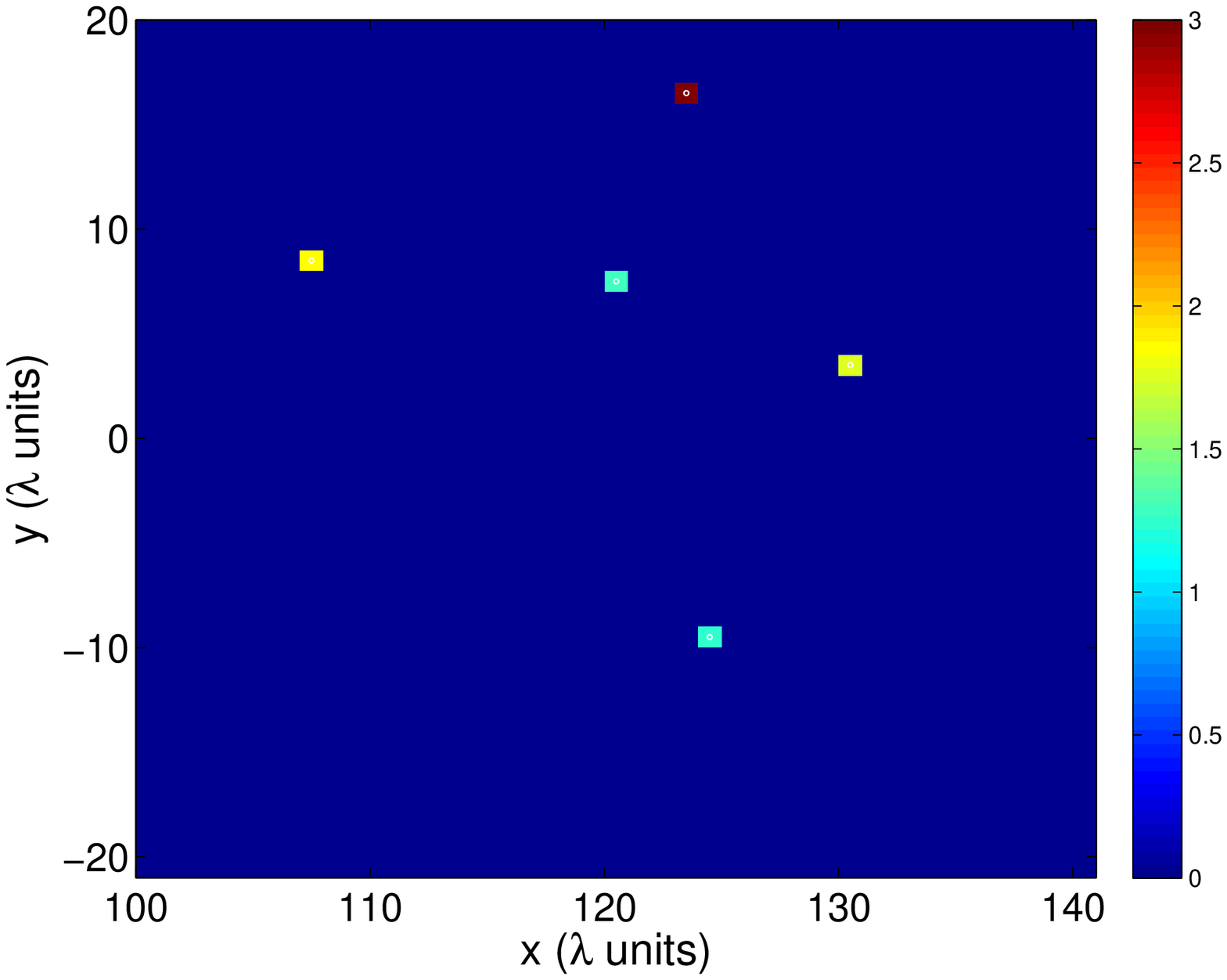} & 
\includegraphics[scale=0.25]{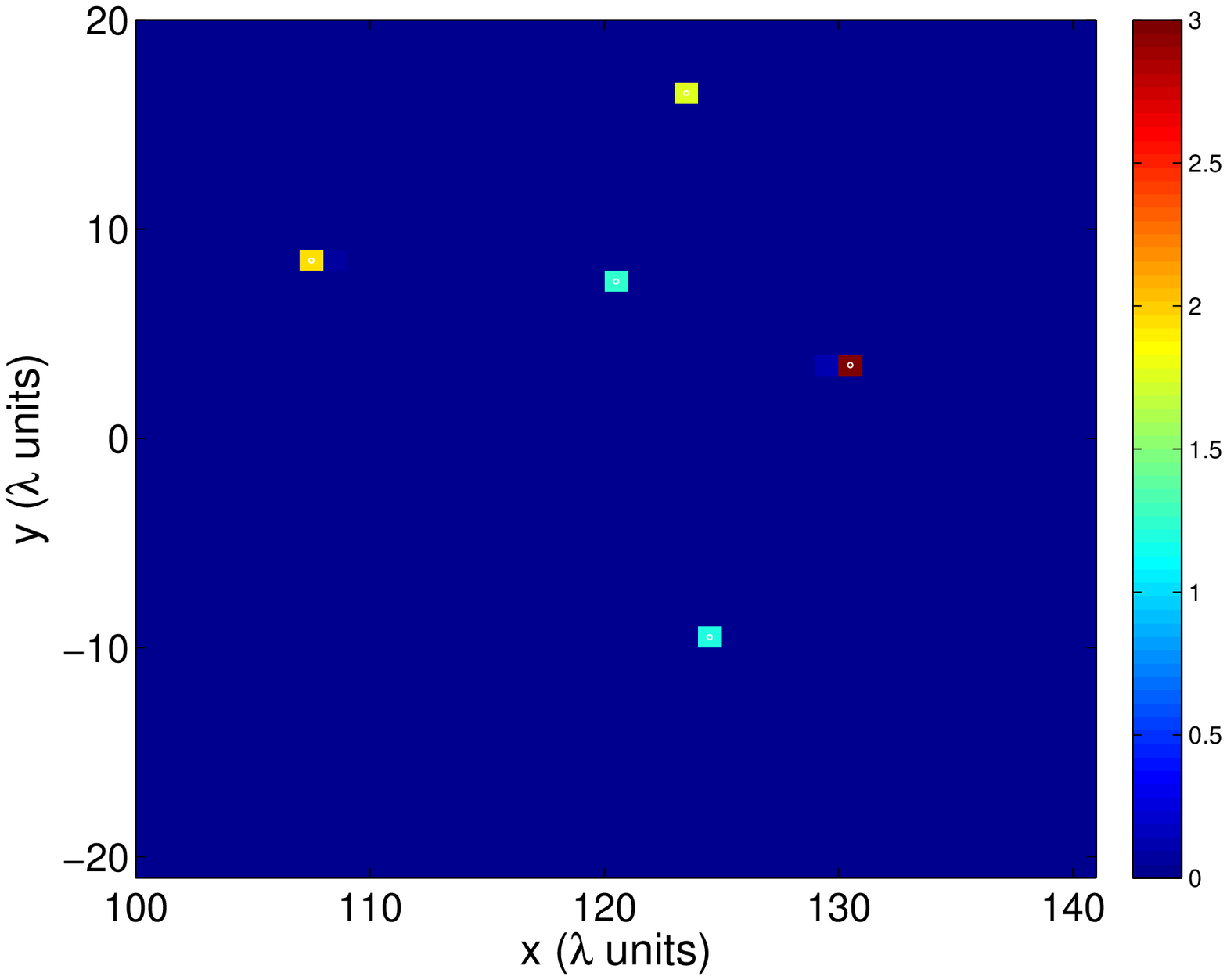} & 
\includegraphics[scale=0.25]{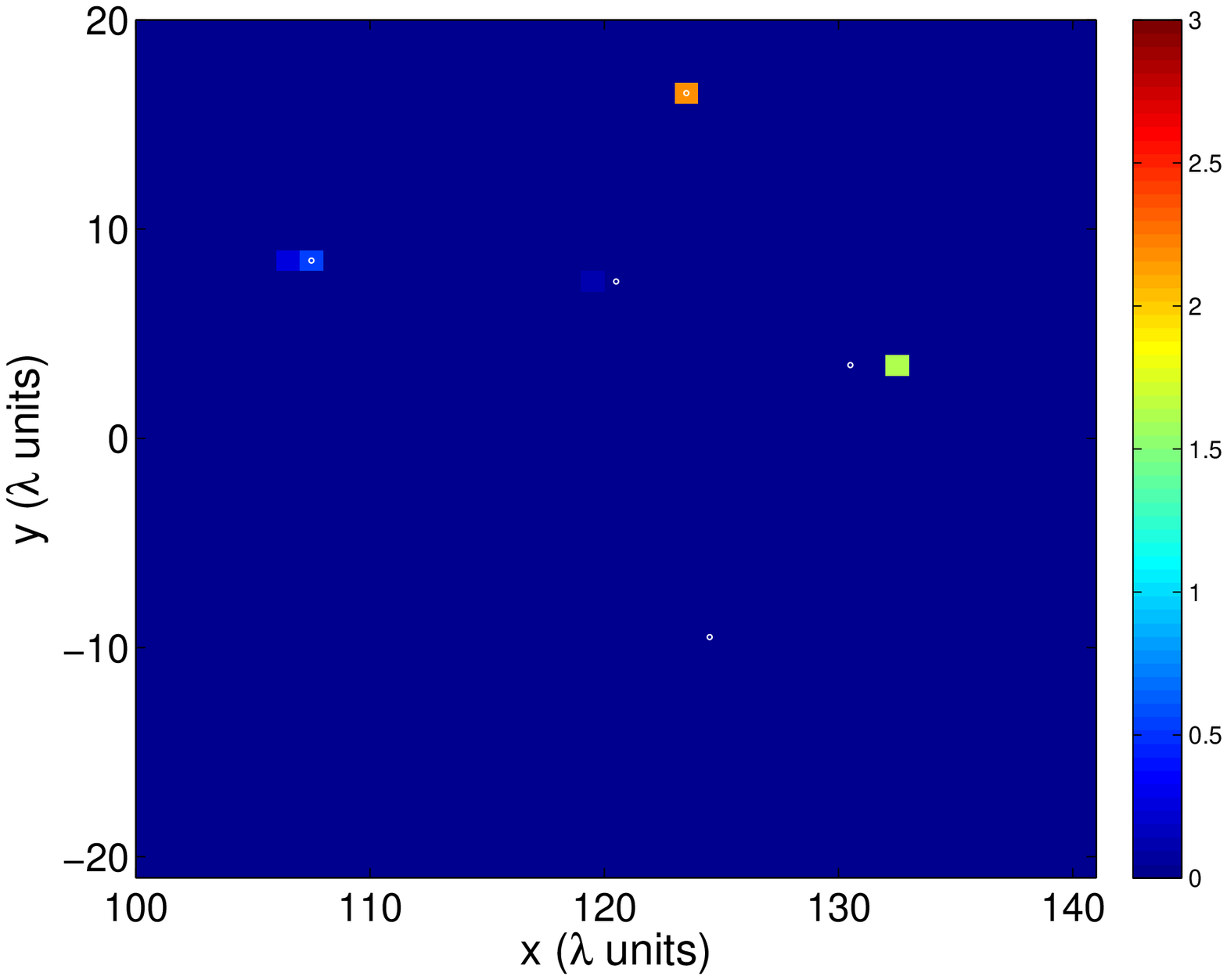} \\
\includegraphics[scale=0.25]{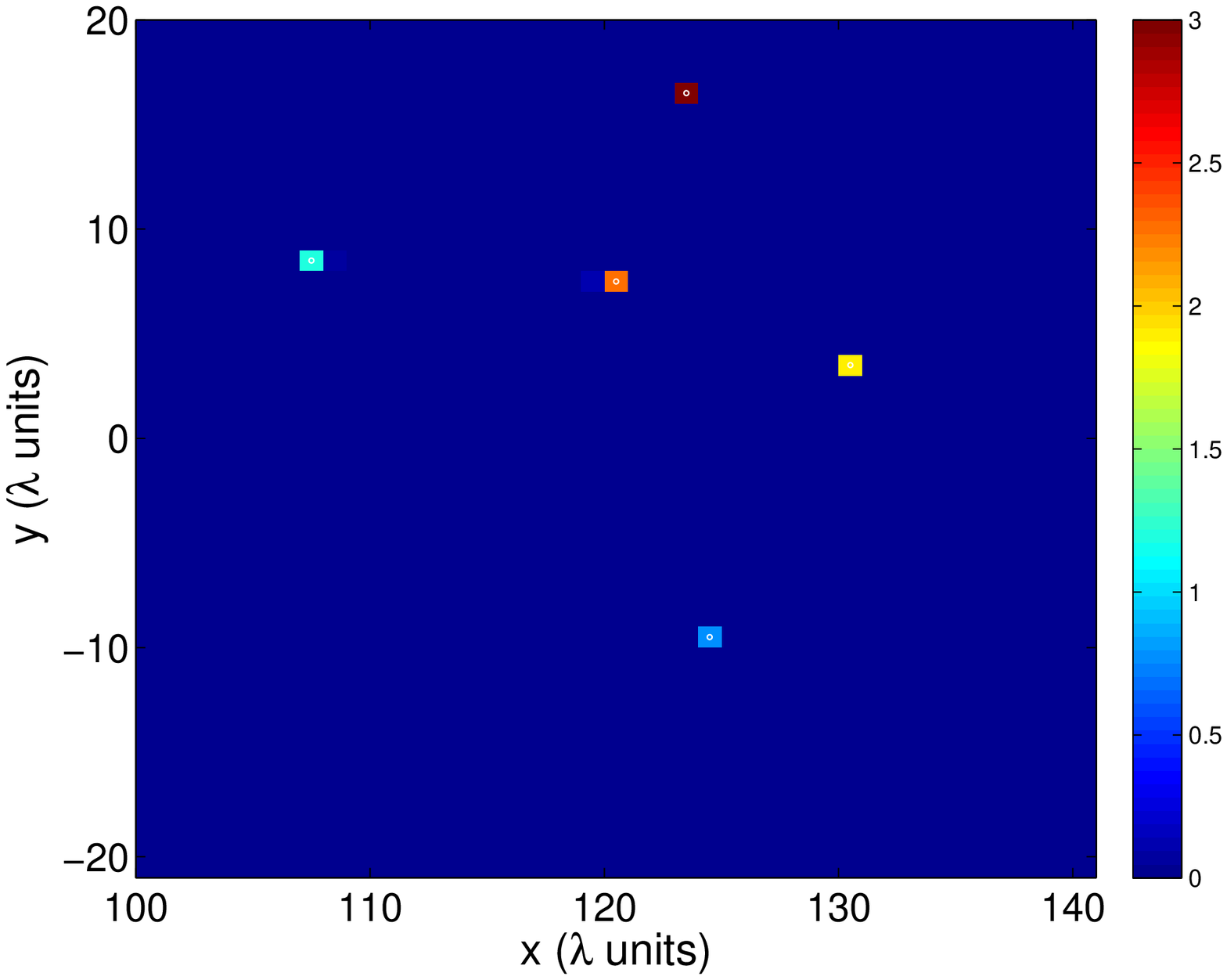} & 
\includegraphics[scale=0.25]{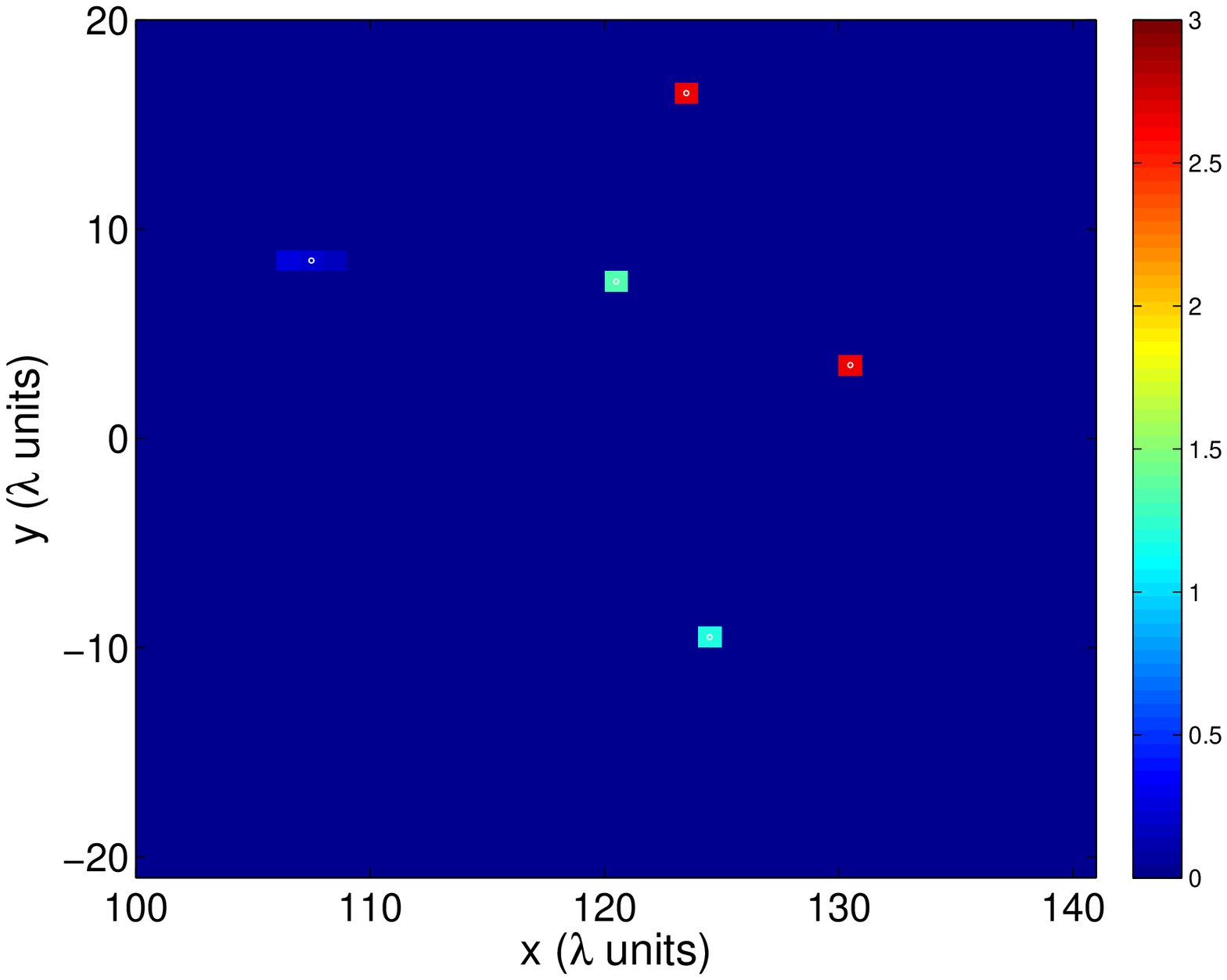} & 
\includegraphics[scale=0.25]{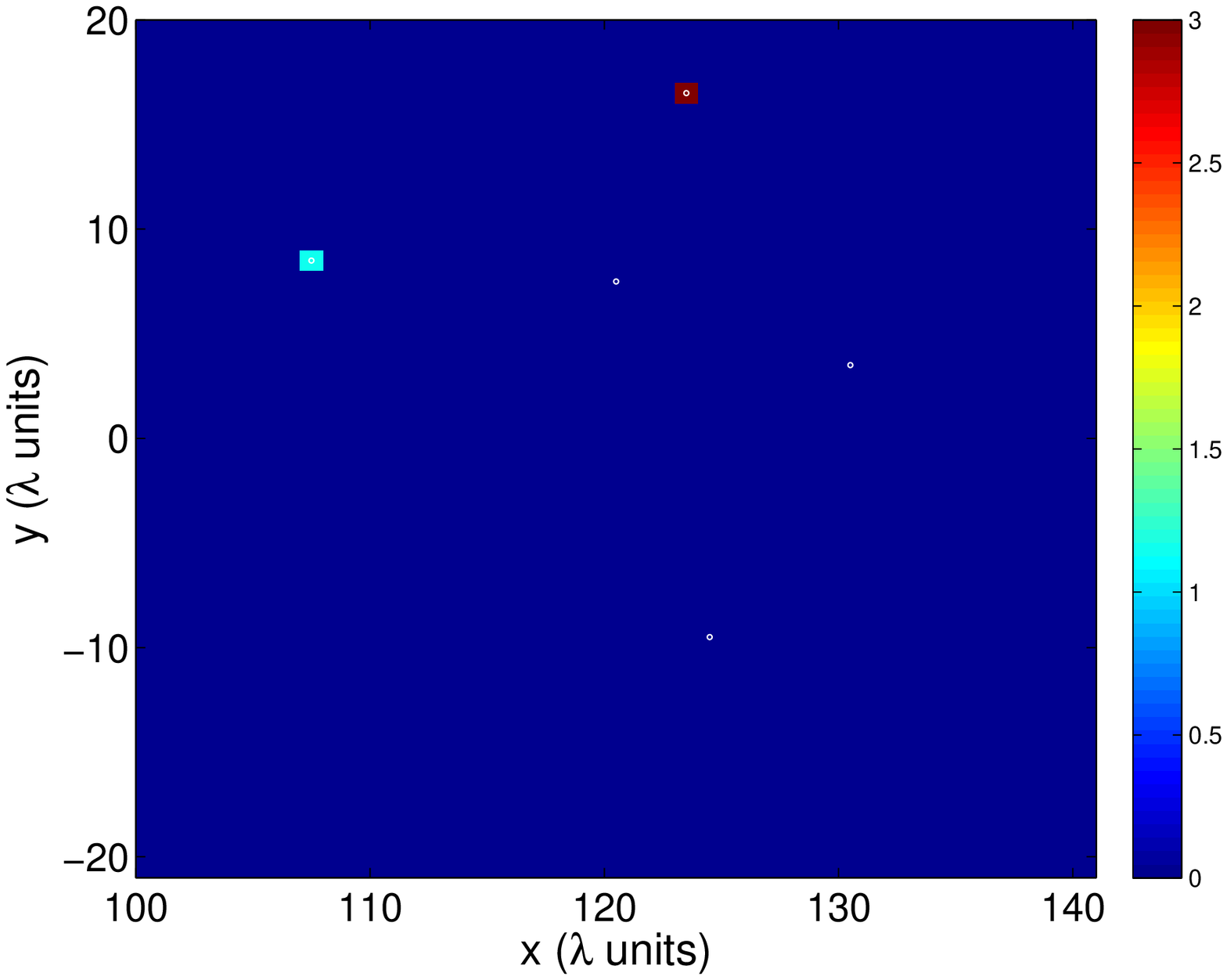}
\end{tabular}
\caption{Images reconstructed by solving \eqref{MMV21noise} when $5$ (top row) and $15$ (bottom row) random
illuminations are used. From left to right, there is  $10\%$, $20\%$, and $50\%$ noise in the data.
}
\label{MMV}
\end{figure}

Figure~\ref{MMV} indicates the importance of selecting ``good" illuminations in the MMV formulation and suggests the use of
optimal illuminations, especially when the signal-to-noise ratio (SNR) is low.
Using optimal illuminations means taking $\vect\wf^{j}=\vect\wV_{\cdot j}$, $j=1,\ldots,M$, as illuminations.
These vectors can be obtained through the SVD of the array response matrix $\vect\wP$ or by iterative time reversal.
Note that, by choosing the illuminations optimally, we (i) maximize the data diversity, which also means low unnecessary redundancy of the multiple illuminations; and (ii) we
reduce the noisy terms contained in the data matrix $\mathbf{B}$.
On the other hand, we point out that, in principle, this approach would require the prior knowledge of the number of scatterers $M$ if the
noise level is high and is difficult to determine the singular values that correspond to the signal space. Hence, it is important to investigate the robustness of this approach with respect to the number of optimal illuminations used
in the scheme. In Figure~\ref{MMVoptimal} we display the results when an increasing number of optimal illuminations are used.
From left to right, and from top to bottom, we use $1$, $2$, $3$, $4$, $5$, and $12$ illuminations associated to the corresponding
singular vectors $\vect\wV_{\cdot j}$, with $j=1,2,3,4,5,12$. We observe that this approach is very robust with respect to the
number of optimal illuminations used. It is remarkable that only a few of them ($2$ or $3$) are enough to achieve a 
significant improvement.
Furthermore, using many more singular vectors as illuminations does not deteriorate too much the images (see the right image in the bottom row,
where $12$ illuminations are use). Finally, we point out that when multiple scattering is negligible all the significant singular
vectors are necessary as shown in \cite{CMP13}. In that case, each optimal illumination aims at one scatterer at a time, provided
that the array is large enough.

\begin{figure}
\centering
\begin{tabular}{ccc}
\includegraphics[scale=0.25]{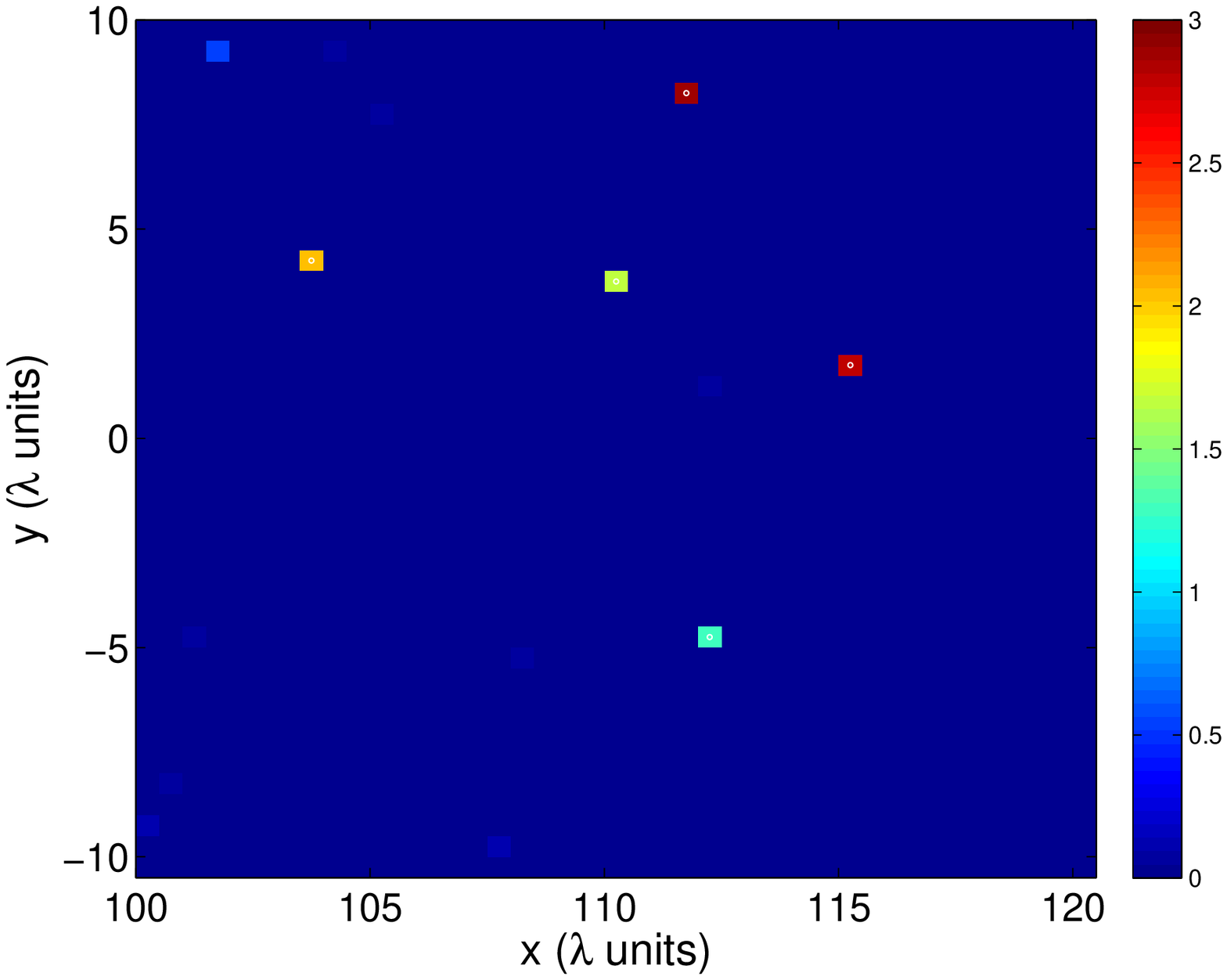} & 
\includegraphics[scale=0.25]{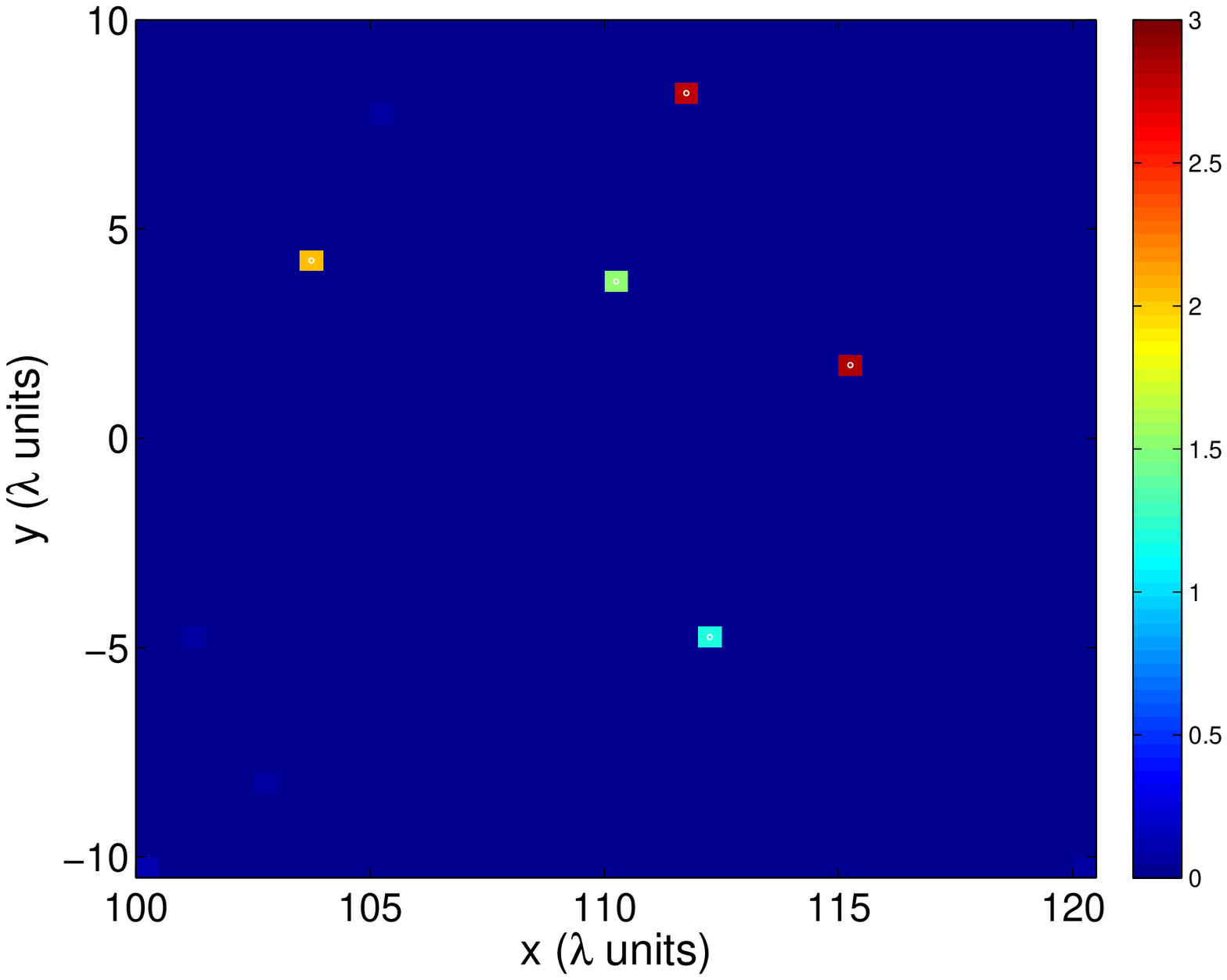} & 
\includegraphics[scale=0.25]{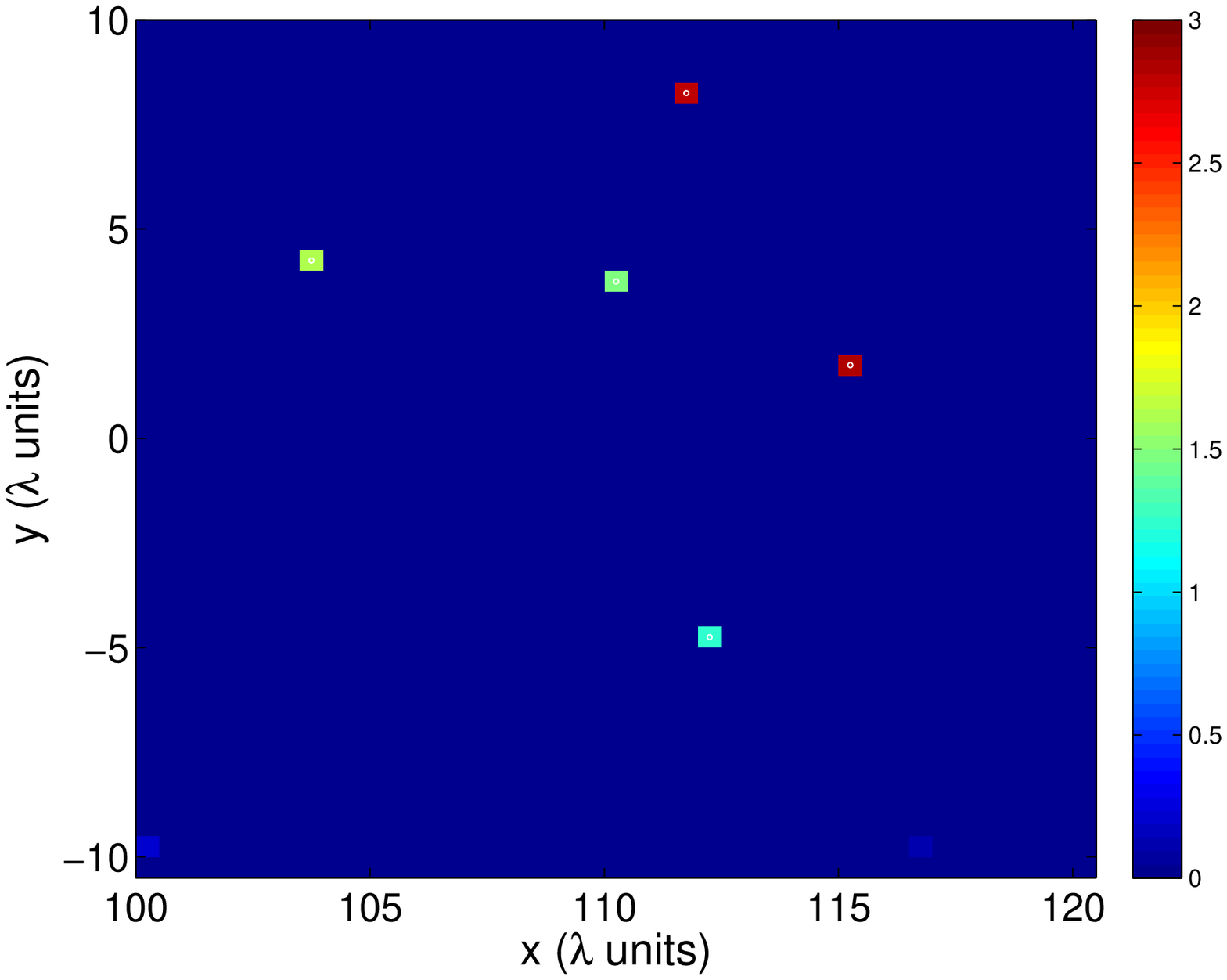}\\
\includegraphics[scale=0.25]{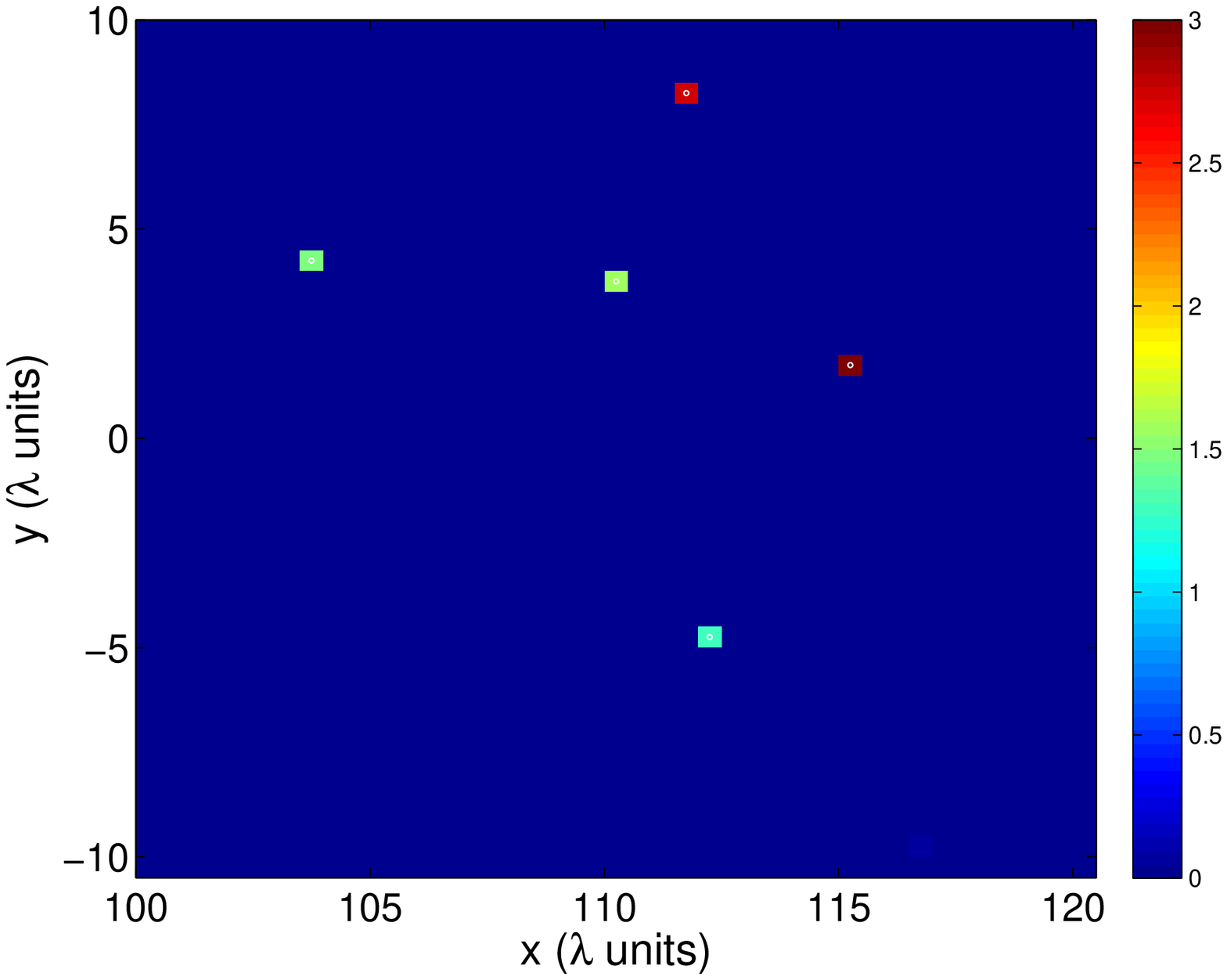} &
 \includegraphics[scale=0.25]{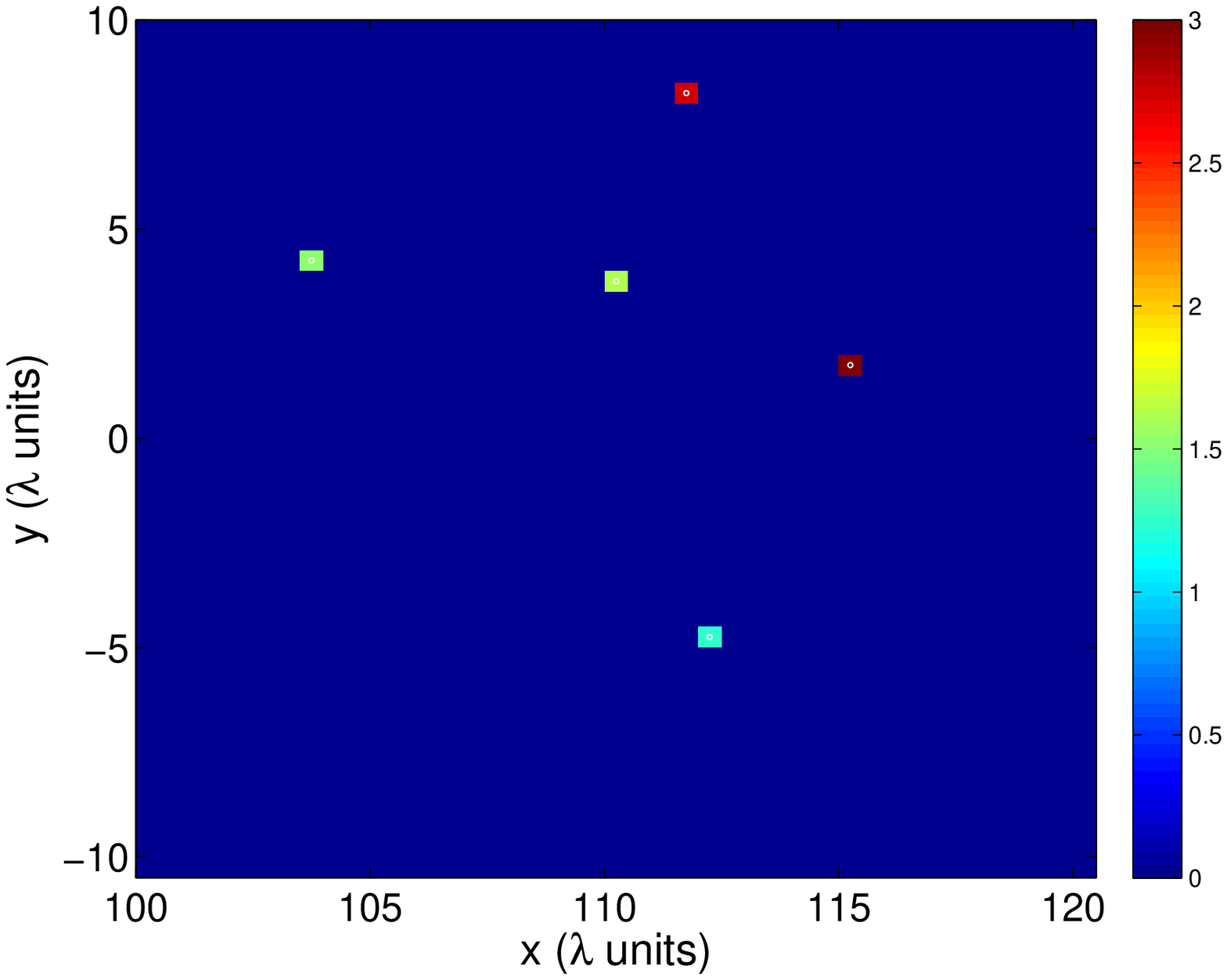} & 
 \includegraphics[scale=0.25]{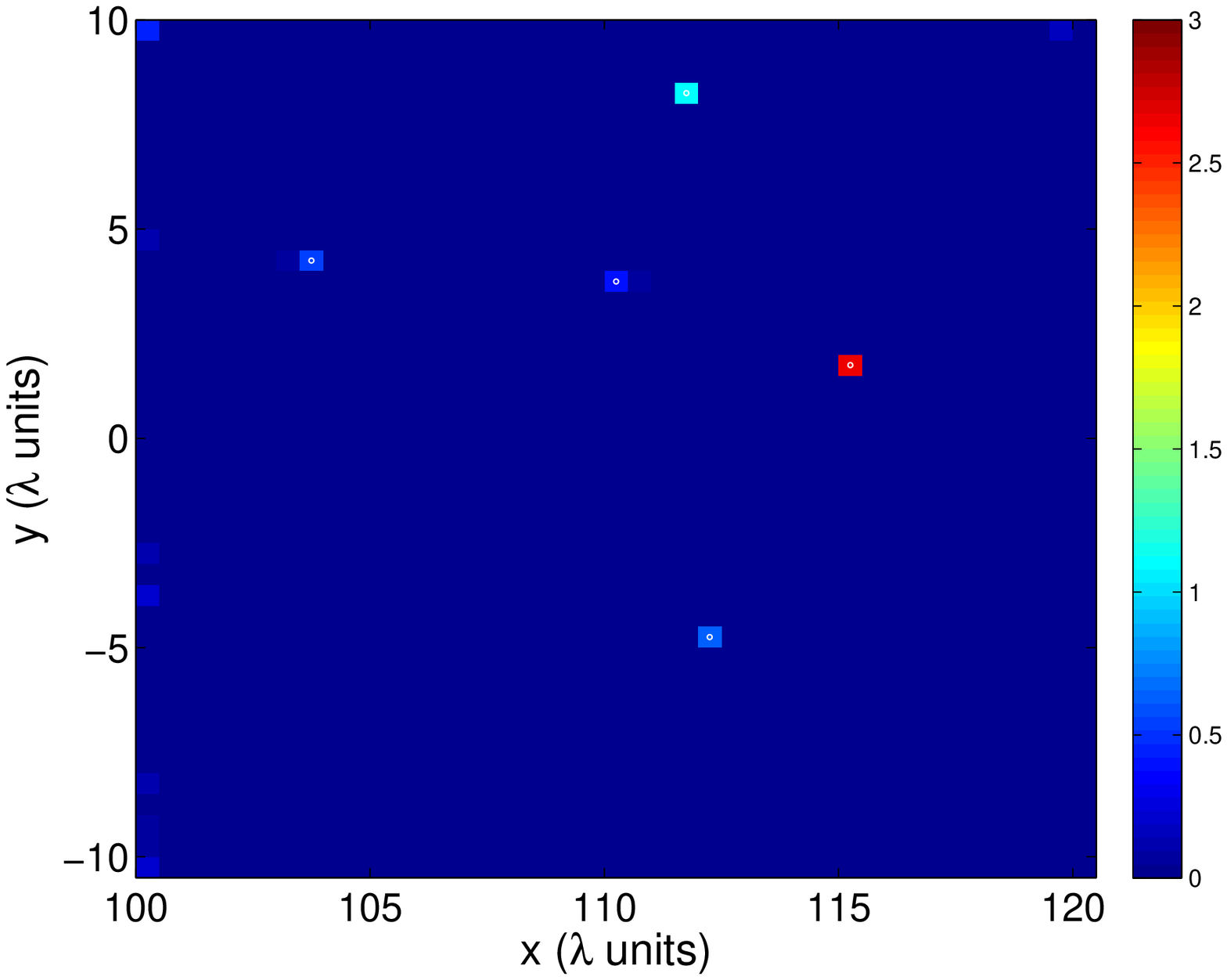}
\end{tabular}
\caption{Images reconstructed by solving \eqref{MMV21noise} when optimal illuminations are used.
There is $50\%$ noise in the data.
From left to right and top to bottom, images are reconstructed by using $1$, $2$, $3$, $4$, $5$, and $12$ top singular vectors.
}
\label{MMVoptimal}
\end{figure}

\begin{figure}
\centering
\begin{tabular}{ccc}
\includegraphics[scale=0.25]{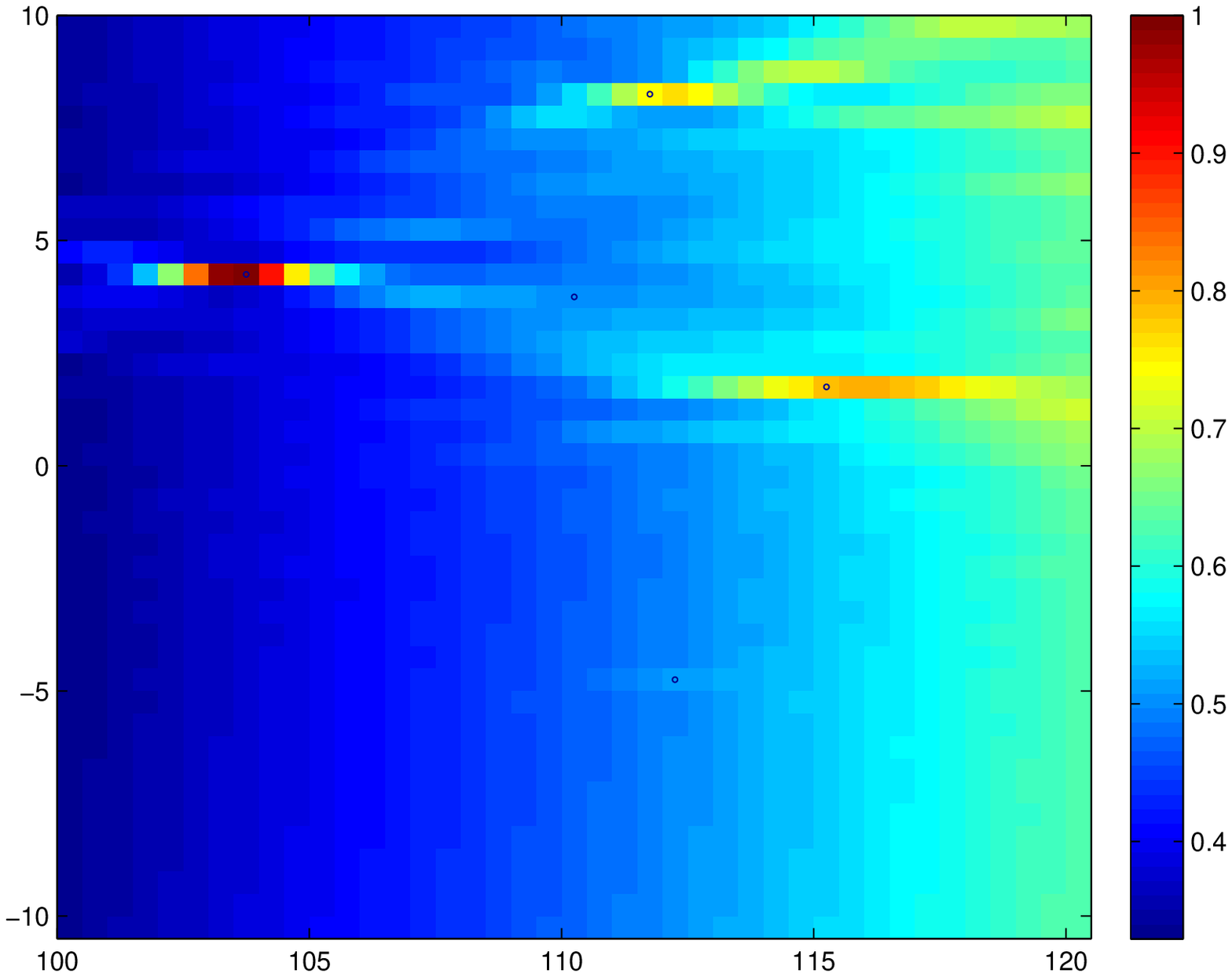} & 
\includegraphics[scale=0.25]{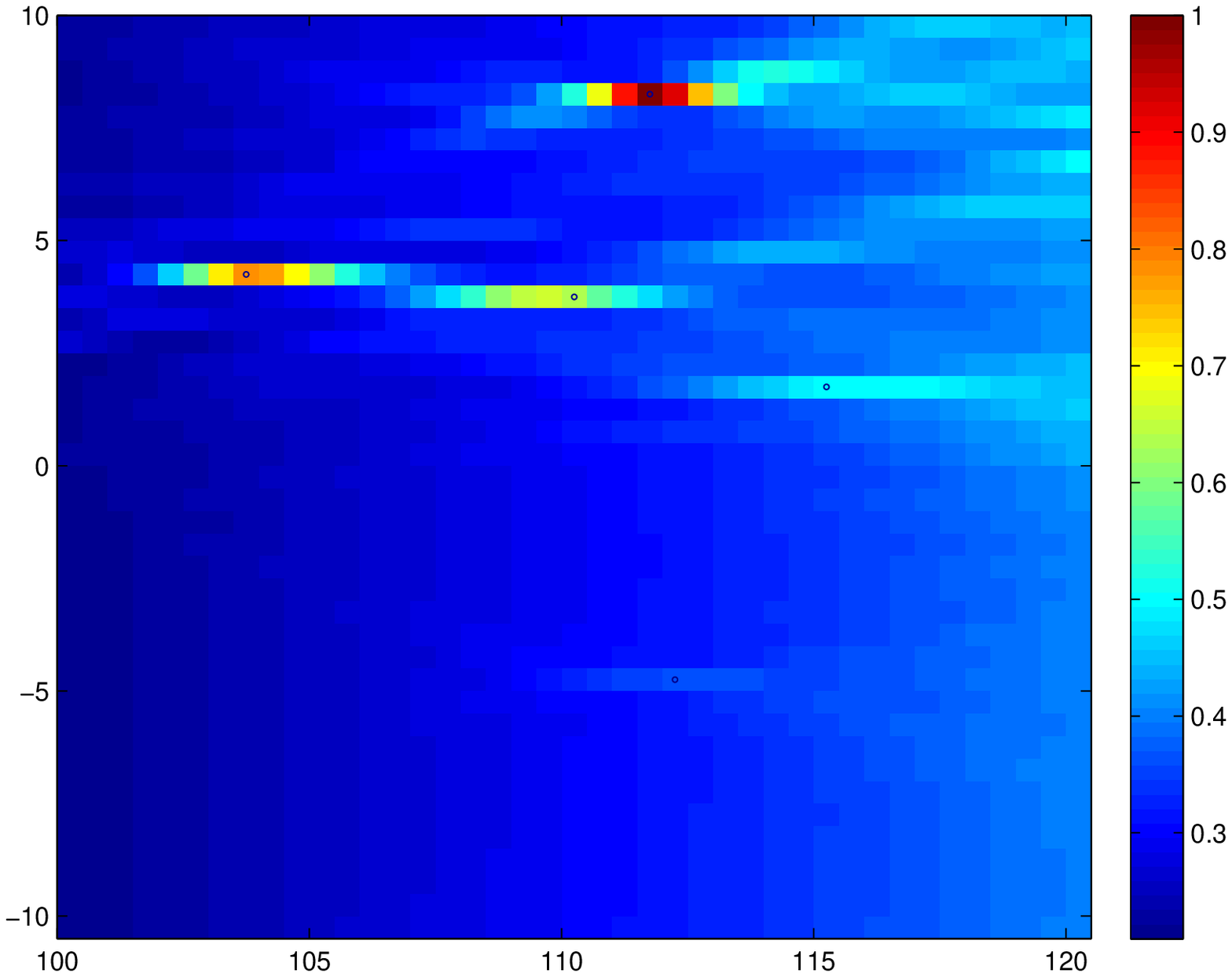} & 
\includegraphics[scale=0.25]{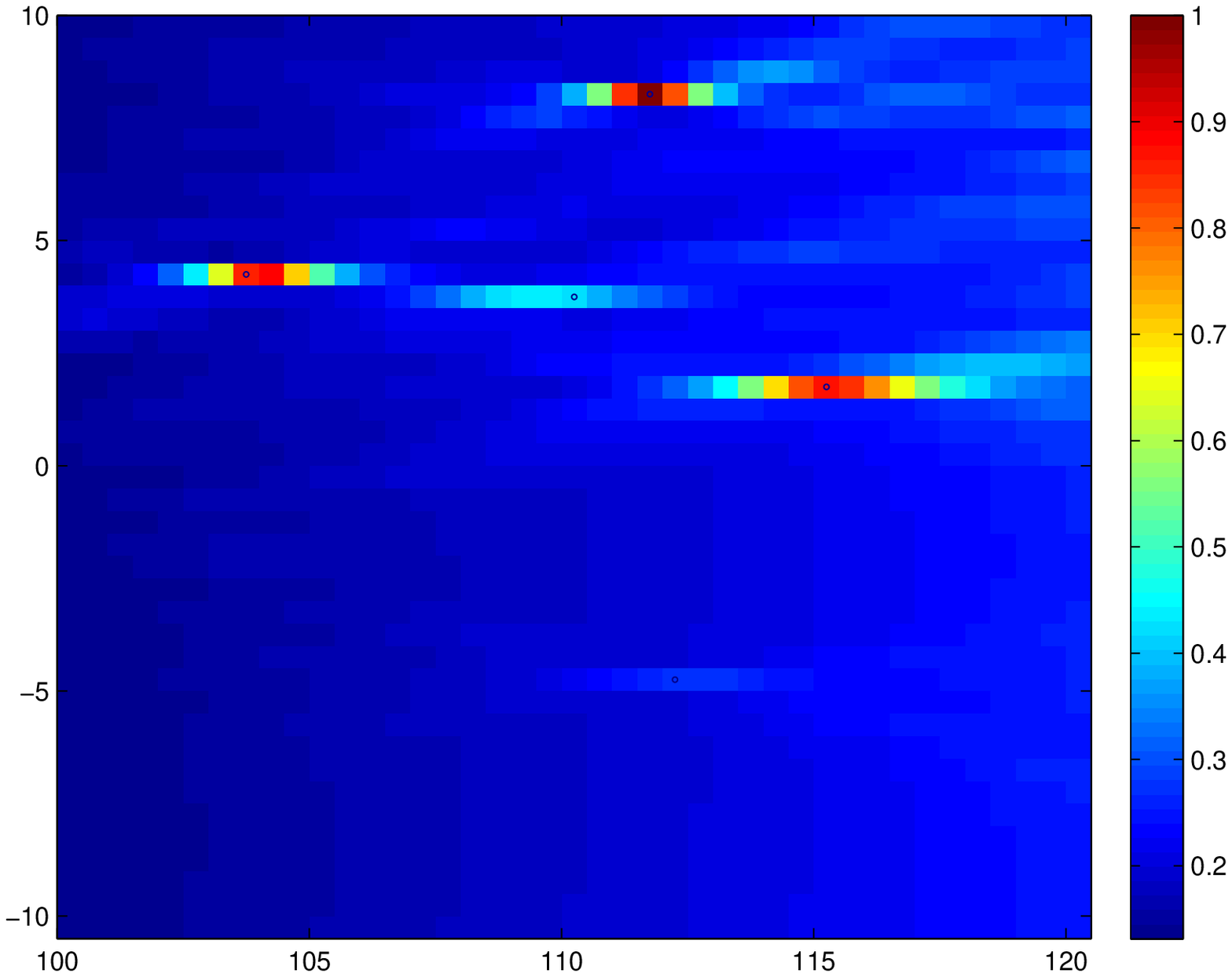}\\
\includegraphics[scale=0.25]{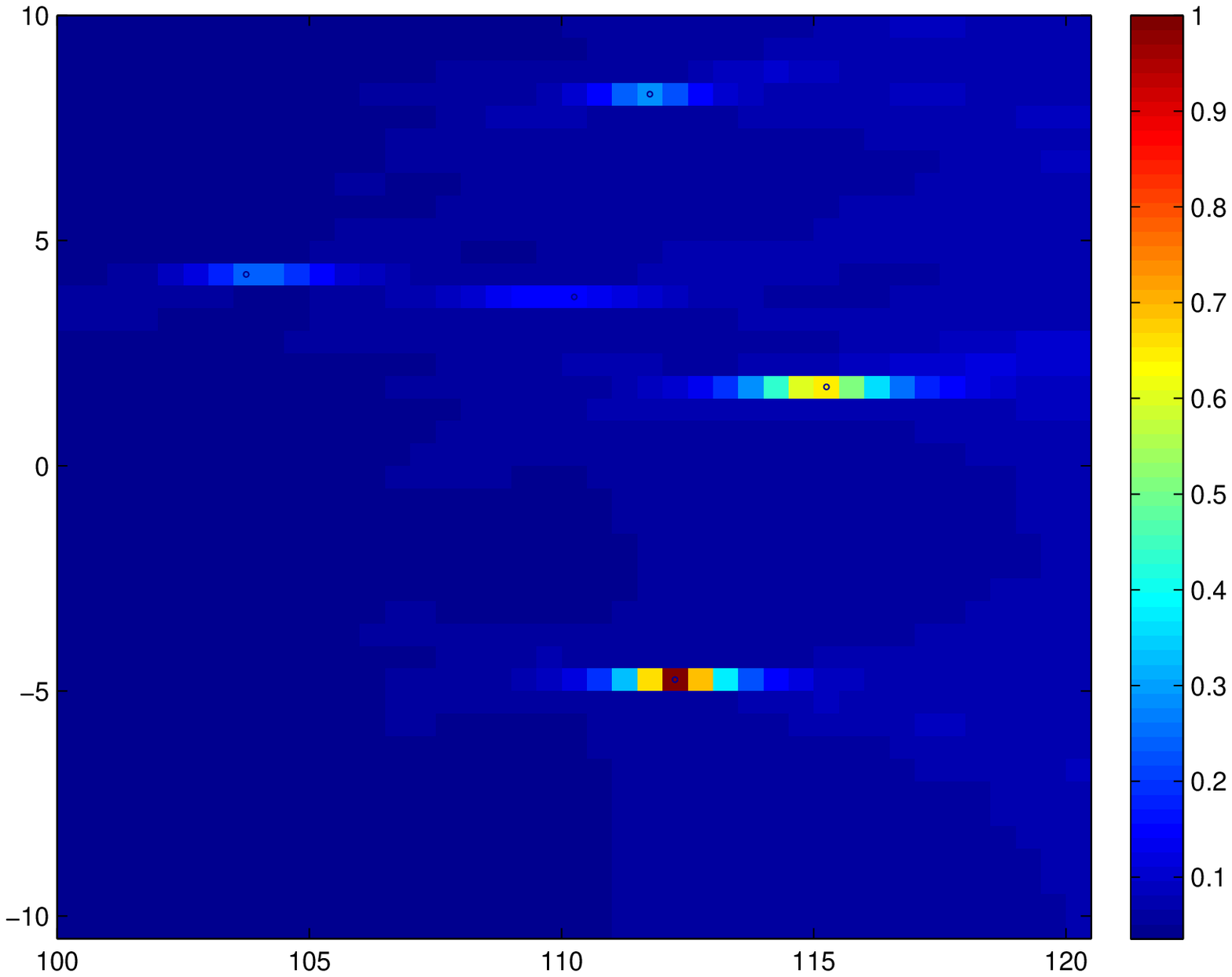} &
 \includegraphics[scale=0.25]{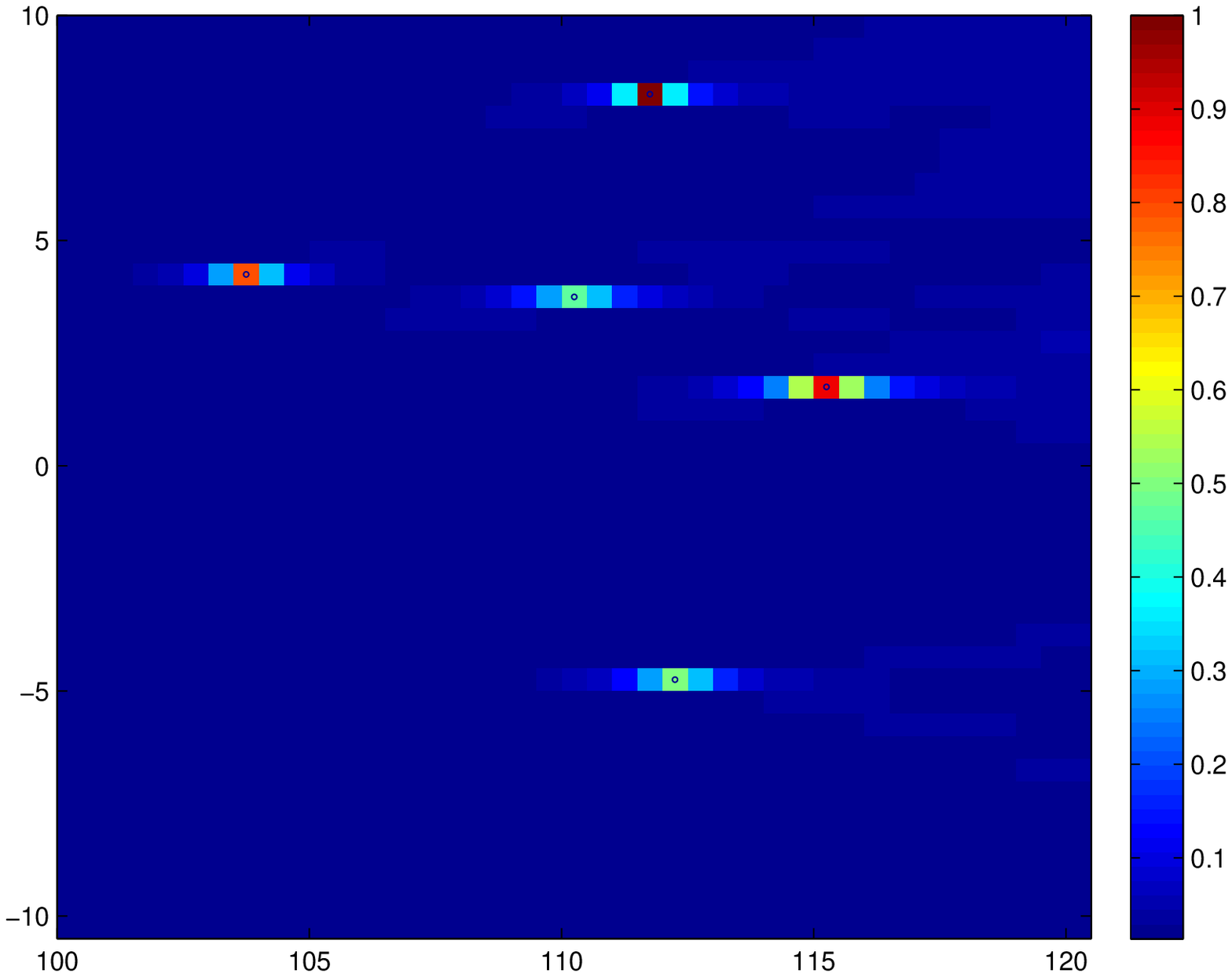} & 
 \includegraphics[scale=0.25]{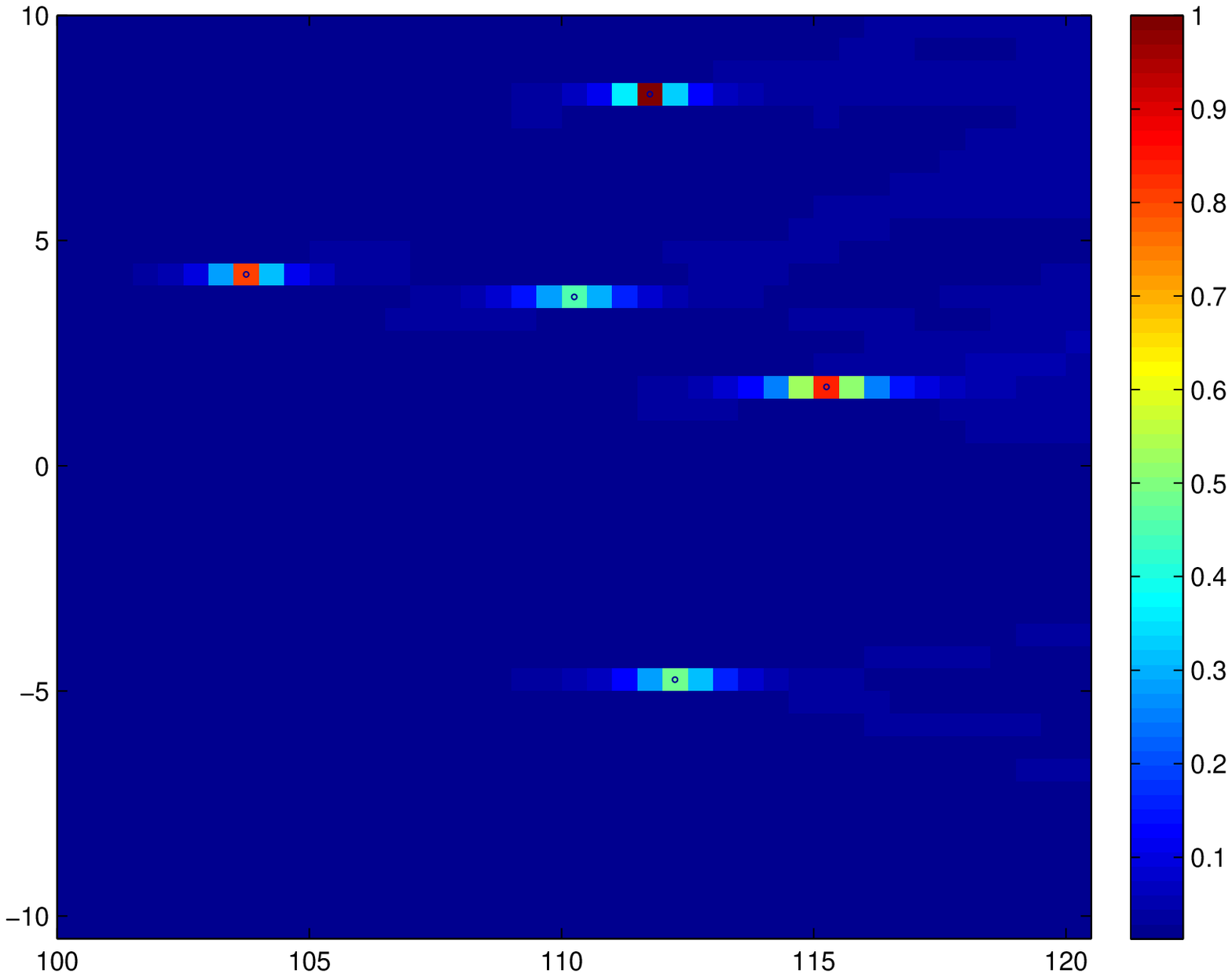}
\end{tabular}
\caption{Images reconstructed by MUSIC. There is $50\%$ noise in the data.
From left to right and top to bottom, images are reconstructed by using $1$, $2$, $3$, $4$, $5$, and $12$ top singular vectors. 
}
\label{MUSIC}
\end{figure}

We compare the images obtained with the MMV formulation and optimal illuminations, with those obtained
with MUSIC. These two methods assume knowledge of the significant singular vectors of the response matrix
$\vect\wP$ to form the images. Thus, the comparison between these two method is carried out with about the
same amount of data. In Figure~\ref{MUSIC} we show the images obtained with MUSIC when an increasing number
of significant singular vectors of the response matrix $\vect\wP$ are used. As expected, we observe that
MUSIC needs to use a number of singular vectors equal or greater to the number of scatterers. This is so,
because MUSIC is a subspace based imaging technique that needs all the significant singular vectors to
span the signal spaced. In other words, the complementary space representing the noise has to be correctly
constructed such that no true signals should fall into. We also observe that the images formed with MUSIC 
do not change much when the number of singular vectors used exceeds the number of scatterers. 
These differences between MMV and MUSIC when both use the essential data of the response matrix $\vect\wP$
is also observed in \cite{MCW05} in the context of DOA for the localization of sources with sensor arrays. 

\section{Conclusion}
We give a novel approach to imaging localized scatterers with non-negligible multiple
scattering between them. Our approach is non-iterative and solves the problem in  two steps using sparsity promoting optimization.
The uniqueness and stability of the formulations using both single and multiple illuminations are analyzed.
We also propose to apply optimal illuminations to improve the robustness of the imaging methods and the
resolution of the images. We show that the conditions under which the proposed methods work well
are related to the configuration of the imaging problems: spherical arrays are in general more favorable
than planar arrays. We illustrated the theoretical results with various numerical examples.

\section*{Appendix}
\begin{appendices}

\section{Proof of Theorem~\ref{thm.mmv}}\label{proof0}
We prove Theorem~\ref{thm.mmv} by proving a more general result given below.
\begin{thm}\label{thm.mmvp}
For a given array configuration, assume that the resolution of the IW satisfies \eqref{mutualcoherence}.
If the number of scatterers $M$ is such that $M\epsilon<1/2$, then
$\vect\Rho_0=[\bfgamma_0^1\,\ldots\,\bfgamma_0^\nu]$
is the unique solution to
\[\min J_{p,1}(\vect\Rho)\quad\text{s.t.}\quad\vect\Gc\vect\Rho=\mathbf{B}\]
for any $1<p<\infty$.
\end{thm}
Clearly, Theorem~\ref{thm.mmv} is a special case of Theorem~\ref{thm.mmvp}. The proof of Theorem~\ref{thm.mmvp} is an
application of the following result which is a generalization of Theorem~$3.1$ in \cite{ER10}.
\begin{prop}\label{prop:lpcondition}
Let $\Lambda$ be the set of row support of $\vect\Rho_0$, i.e.
$\Lambda=\operatorname{rowsupp}(\vect\Rho_0)$. For any $1<p<\infty$ and a matrix $\mathbf{Q}$, define
\[
\operatorname{sign}(Q_{ij})=\begin{cases}
\frac{|Q_{ij}|^{p-1}\operatorname{sign}(Q_{ij})}{\||Q_{i\cdot}|^{p-1}\|_{\ell_q}},&\|Q_{i,\cdot}\|_{\ell_q}\neq0\\
0,&\|Q_{i,\cdot}\|_{\ell_q}=0,
\end{cases}
\]
where $\operatorname{sign}(x)=0,\pm1$ when $x\in\mR$ and $\operatorname{sign}(x)=\exp(\mathrm{i}\operatorname{angle}(x))$ when $x\in\mC$.
Also assume $\vect\Gc_{\Lambda}$, submatrix of $\vect\Gc$ consisting of columns with indices in $\Lambda$,
is non-singular. Then a sufficient condition under which $\vect\Rho_0$ is the
unique solution to \[\min_{\vect\Rho}J_{p,1}(\vect\Rho)\quad\text{s.t.}\quad\vect\Gc\vect\Rho=\mathbf{B}\]
is that there exists a matrix $\mathbf{H}\in\mC^{N\times\nu}$ satisfying
\begin{equation}\label{lpcondition1}
\vect\Gc_{\Lambda}^\ast\mathbf{H}=\operatorname{sign}(\vect\Rho_{0\Lambda})
\end{equation}
and 
\begin{equation}\label{lpcondition2}
\|\mathbf{H}^\ast\vect\wg_0(\vect y_j)\|_{\ell_q}<1,
\end{equation}
where $1/p+1/q=1$ and $\vect\Rho_{0\Lambda}$ is the submatrix consisting of the rows of $\vect\Rho_0$ in $\Lambda$.
\end{prop}
To prove Proposition~\ref{prop:lpcondition}, we need the following lemma.
\begin{lem}\label{lem:traceholder}
For any two matrices $A\in\mC^{m\times l}$ and $B\in\mC^{l\times n}$, we have
\[\big|\operatorname{trace}(AB)\big|\le\max_{k=1,\ldots,l}\|B_{\cdot k}\|_{\ell_q}\,J_{p,1}(A).\]
The strict inequality holds when there exists $k$ such that $\|B_{\cdot k}\|_{\ell_q}<\max_{k=1,\ldots,l}\|B_{\cdot k}\|_{\ell_q}$
and $\|A_{k \cdot}\|_{\ell_p}\neq0$.
\end{lem}
\begin{proof}
By definition of trace, we have
\begin{equation*}
|\operatorname{trace}(AB)|\le\sum_{k=1}^l|A_{k \cdot}B_{\cdot k}|\le\sum_{k=1}^l\|A_{k \cdot}\|_{\ell_p}\|B_{\cdot k}\|_{\ell_q}\le\max_{k=1,\ldots,l}\|B_{\cdot k}\|_{\ell_q}\,J_{p,1}(A)
\end{equation*}
where we use the H\"{o}lder's inequality in the second to last inequality. The strict inequality apparently holds when the condition is satisfied.
\end{proof}
\begin{proof}[Proof of Proposition~\ref{prop:lpcondition}]
We will show the uniqueness by contradiction. Assume there exists another solution $\widehat{\vect\Rho}$ the
support of which is $\widehat{\Lambda}$ such that $\widehat{\Lambda}\backslash\Lambda\neq\emptyset$.
First of all, notice that \eqref{lpcondition1}, \eqref{lpcondition2} imply that for any column of $\vect\Gc$,
$\|\mathbf{H}^\ast\vect\wg_0(\vect y_j)\|_{\ell_q}\le1$. We have
\[J_{p,1}(\vect\Rho_0)=J_{p,1}(\vect\Rho_{0\Lambda})
=\operatorname{trace}(\operatorname{sign}(\vect\Rho_{0\Lambda})\vect\Rho_{0\Lambda}^\ast)
=\operatorname{trace}(\vect\Gc_{\Lambda}^\ast\mathbf{H}\vect\Rho_{0\Lambda}^\ast).
\]
Since the $\operatorname{trace}$ function is invariant with respect to matrix rotation and transpose operation,
\[J_{p,1}(\vect\Rho_0)
=\operatorname{trace}(\vect\Rho_{0\Lambda}\mathbf{H}^\ast\vect\Gc_{\Lambda})
=\operatorname{trace}(\mathbf{H}^\ast\vect\Gc_{\Lambda}\vect\Rho_{0\Lambda})
=\operatorname{trace}(\mathbf{H}^\ast\vect\Gc\widehat{\vect\Rho}).\]
Applying Lemma~\ref{lem:traceholder}, we have
\[J_{p,1}(\vect\Rho_0)\le\max_{j\in\widehat{\Lambda}}\|\mathbf{H}^\ast\vect\wg_0(\vect y_j)\|_{\ell_q}\,J_{p,1}(\widehat{\vect\Rho})\le J_{p,1}(\widehat{\vect\Rho}).\]
Because $\widehat{\Lambda}\backslash\Lambda\neq\emptyset$, there must exist $j_0\in\widehat{\Lambda}$ such that
$\|\mathbf{H}^\ast\vect\wg_0(\vect y_{j_0})\|_{\ell_q}\neq\max_{j\in\widehat{\Lambda}}\|\mathbf{H}^\ast\vect\wg_0(\vect y_j)\|_{\ell_q}$. On the other hand, since the support of $\widehat{\vect\Rho}$ is $\widehat{\Lambda}$, we have
$\|\widehat\Rho_{j\cdot}\|_{\ell_q}\neq0$ for any $j\in\widehat{\Lambda}$. According to Lemma~\ref{lem:traceholder}, $J_{p,1}(\vect\Rho_0)<J_{p,1}(\widehat{\vect\Rho})$ which contradicts that $\widehat{\vect\Rho}$ is also a solution. Therefore, the solution must be unique.
\end{proof}

Now we will show that the multiplier $\mathbf{H}$ satisfying \eqref{lpcondition1} and \eqref{lpcondition2}
exists under the condition of Theorem~\ref{thm.mmvp}.
\begin{proof}[Proof of Theorem~\ref{thm.mmvp}]
Let $\Lambda=\{n_j, 1\le j\le M\}$ be the set of indices corresponding to the scatterers. Based on the
resolution condition, we have that the inner product of the column vectors of the matrix $\vect\Gc$ satisfies
$\vect\wg_0^\ast(\vect y_i)\vect\wg_0(\vect y_j)=\delta_{ij}+(1-\delta_{ij})\epsilon_{ij}$
with $|\epsilon_{ij}|<\epsilon$, for any $1\le i,j\le M$. Therefore, the submatrix $\vect\Gc_\Lambda$,
composed of the columns $n_1,\ldots,n_M$ of matrix $\vect\Gc$, is full column rank and
satisfies that $\vect\Gc_\Lambda^\ast\vect\Gc_\Lambda$ is full rank and diagonally dominant. 

According to Proposition~\ref{prop:lpcondition}, we need to find a matrix $\mathbf{H}$ satisfying
\eqref{lpcondition1} and \eqref{lpcondition2}.
Let $\mathbf{H}=\vect\Gc_\Lambda(\vect\Gc_\Lambda^\ast\vect\Gc_\Lambda)^{-1}\operatorname{sign}(\vect\Rho_{0\Lambda})$. Then, the first condition is automatically satisfied because $\vect\Gc_\Lambda^\ast\mathbf{H}= \operatorname{sign}(\vect\Rho_{0\Lambda})$. For the second condition, choosing any column $j$ of $\vect\Gc$ not in the submatrix $\vect\Gc_\Lambda$, we have
\begin{eqnarray*}
\|\vect\wg_0^\ast(\vect y_j)\vect\Gc_\Lambda(\vect\Gc_\Lambda^\ast\vect\Gc_\Lambda)^{-1}\operatorname{sign}(\vect\Rho_{0\Lambda})\|_{\ell_q}&\le&\|\operatorname{sign}(\vect\Rho_{0\Lambda})\|_{p\rightarrow\infty}\|(\vect\Gc_\Lambda^\ast\vect\Gc_\Lambda)^{-1}\|_{\ell_1}\|\vect\Gc_\Lambda^\ast\vect\wg_0(\vect y_j)\|_{\ell_1}\\
&\le&\frac{M\epsilon}{1-M\epsilon}<1,
\end{eqnarray*}
where by definition, we have $\|\operatorname{sign}(\vect\Rho_{0\Lambda})\|_{p\rightarrow\infty}\le1$ and the last inequality
is due to the resolution condition $M\epsilon<1/2$.
\end{proof}

\section{Proof of Theorem~\ref{thm.mmvnoise}}\label{proof}
In \cite{TROPP06-2}, the author gives conditions for the MMV problem under which the convex relaxation formulation with functional $J_{\infty,1}$ is robust with respect to the additive noise. In the following, we derive similar conditions for convex relaxation using $J_{2,1}$ instead of $J_{\infty,1}$.
This is done using the techniques developed in \cite{TROPP04}. We first introduce some supporting results.
\begin{definition}\label{def:subdiff}
Let $f$ be a function from the complex matrix space ${\cal M}$ to $\mC$.
The subdifferential of a function $f$ at $\vect\Rho\in{\cal M}$ is defined as
$$
\partial f(\mathbf{X}) = \{
\mathbf{G} \in {\cal M}: f(\mathbf{Y}) \ge f(\mathbf{X}) + \mbox{Re}\ltr\mathbf{Y}-\mathbf{X},\mathbf{G}\rtr, \,\,\forall \,\,\mathbf{Y} \in {\cal M}
\}\, .
$$
\end{definition}

\begin{lem}\label{lemma:subgradJ21}
A matrix $\mathbf{G}$ lies in the subdifferential of $J_{2,1}(\mathbf{X})$ at $\mathbf{X}\in\mC^{K\times\nu}$ if and only if its rows $G_{i\cdot}\in\mC^\nu$ satisfy
\begin{itemize}
\item $G_{i\cdot}=\frac{X_{i\cdot}}{\|X_{i\cdot}\|_{\ell_2}}$ when $X_{i\cdot}\neq0$, and
\item $\|G_{i\cdot}\|_{\ell_2}\le1$ when $X_{i\cdot}=0$.
\end{itemize}
Moreover, $\mathbf{G}$ is called a subgradient of $J_{2,1}(\mathbf{X})$ at $\mathbf{X}$.
\end{lem}
Then according to the definition of matrix norms, it is easy to see the subgradient of $J_{2,1}$ satisfies the following.
\begin{col}\label{col:subgradientnorm}
Any subgradient $\mathbf{G}$ of $J_{2,1}$ satisfies $\|\mathbf{G}\|_{2\rightarrow\infty}=\|\mathbf{G}^\ast\|_{1,2}\le1$.
\end{col}

We now give a result related to operator norms of matrices which will be used later.
\begin{lem}\label{prop:norm}
For any given matrix $\mathbf{A}$, we have
$\|\mathbf{A}\|_{(2,\infty)\rightarrow F}\le\|\mathbf{A}\|_{\infty\rightarrow2}$.
\end{lem}
\begin{proof}
Let $\mathbf{A}$ be an $m\times n$ matrix. By definition, we have
\[\|\mathbf{A}\|_{(2,\infty)\rightarrow F}=\max_{\mathbf{C}\in\mR^{n\times d}}\frac{\|\mathbf{AC}\|_F}{\|\mathbf{C}\|_{2\rightarrow\infty}}.\]
Since we have
\begin{eqnarray*}
\left(\frac{\|\mathbf{AC}\|_F}{\|\mathbf{C}\|_{2\rightarrow\infty}}\right)^2&=&\frac{\sum_{i=1}^m\sum_{j=1}^d\bigg|\sum_{k=1}^nA_{ik}C_{kj}\bigg|^2}{\max_{1\le i\le n}\|C_{i\cdot}\|_{\ell_2}^2}\\
&\le&\frac{\sum_{i=1}^m\bigg|\sum_{k=1}^nA_{ik}\|C_{k\cdot}\|_{\ell_2}\bigg|^2}{\max_{1\le i\le n}\|C_{i\cdot}\|_{\ell_2}^2}\\
&\le&\max_{\vect c\in\mR^n}\left(\frac{\|\mathbf{A}\vect c\|_{\ell_2}}{\|\vect c\|_\infty}\right)^2\\
&\le&\|\mathbf{A}\|^2_{\infty\rightarrow2},
\end{eqnarray*}
it is clearly that $\|\mathbf{A}\|_{(2\rightarrow\infty)\rightarrow F}\le\|\mathbf{A}\|_{\infty\rightarrow2}$.
In the derivation above, the first inequality is true because for each row index $i$, we have
\begin{eqnarray*}
\sum_{j=1}^d\bigg|\sum_{k=1}^nA_{ik}C_{kj}\bigg|^2&=&\sum_{j=1}^d\sum_{k,k'=1}^nA_{ik}C_{kj}\bar{A}_{ik'}\bar{C}_{k'j}\\
&=&\sum_{k,k'=1}^nA_{ik}\bar{A}_{ik'}\sum_{j=1}^dC_{kj}\bar{C}_{k'j}\\
&=&\sum_{k,k'=1}^nA_{ik}\bar{A}_{ik'}\langle C_{k\cdot}, C_{k'\cdot}\rangle\\
&\le&\sum_{k,k'=1}^nA_{ik}\bar{A}_{ik'}\|C_{k\cdot}\|_{\ell_2}\|C_{k'\cdot}\|_{\ell_2}\\
&=&\left|\sum_{k=1}^nA_{ik}\|C_{k\cdot}\|_{\ell_2}\right|^2.
\end{eqnarray*}
\end{proof}

To make the following discussion easier, we assume that the sensing matrix $\vect\Gc$ has normalized columns and we introduce some additional notations used in this appendix only.
Let $\Lambda$ be indexes of a subset of linearly independent columns of $\vect\Gc$, i.e. $\Lambda\subset\{1,2,\ldots,K\}$ such
that the Green's function vectors $\vect\wg_0(\vect y_j)$, with $j\in\Lambda$, are linearly independent.
We denote by $\vect\Gc_{\Lambda}\in\mC^{N\times|\Lambda|}$
the submatrix of $\vect\Gc$ composed of columns with indices in $\Lambda$, by $\mathbf{B}_{\Lambda}\in\mC^{N\times\nu}$
the best Frobenius norm approximation of the data matrix $\mathbf{B}$ over $\Lambda$, and by $\mathbf{X}_{0\Lambda}\in\mC^{|\Lambda|\times\nu}$ the
corresponding coefficent matrix synthesizing $\mathbf{B}_{\Lambda}$, i.e. 
such that $\mathbf{B}_{\Lambda}=\vect\Gc_{\Lambda}\mathbf{X}_{0\Lambda}$.
Note that $\mathbf{X}_{0\Lambda}=\vect\Gc_{\Lambda}^\dag\mathbf{B}_{\Lambda}$, with
$\vect\Gc_{\Lambda}^\dag=(\vect\Gc_{\Lambda}^\ast\vect\Gc_{\Lambda})^{-1}\vect\Gc_\Lambda^\ast$.

Next, we give several results related to the minimizers of the Lagrange function \eqref{MMV-functional}.
The proofs are analogue to those in \cite{TROPP06-2} and will be skipped. Interested readers can refer
to \cite{TROPP06-2} or \cite{TROPP04} for single measurement case. The first lemma is on the condition of the minimizer of \eqref{MMV-functional}.
\begin{lem}\label{lemma:lagrangecondition}
Suppose that the matrix $\mathbf{X}_\star$ is the minimizer of \eqref{MMV-functional} over all
matrices with row-support $\Lambda$. A necessary and sufficient condition for
$\mathbf{X}_\star$ to be such minimizer is that
\begin{equation}\label{eq:uniquelagrangecondition}
\mathbf{X}_{0\Lambda}-\mathbf{X}_\star=\lambda(\vect\Gc_{\Lambda}^\ast\vect\Gc_{\Lambda})^{-1}\mathbf{G},
\end{equation}
where $\mathbf{G}\in\partial J_{2,1}(\mathbf{X}_\star)$. Moreover, the minimizer is unique.
\end{lem}
Using Lemmas~\ref{prop:norm} and ~\ref{lemma:lagrangecondition}, we have the following estimates
on the bound of the minimizer of \eqref{MMV-functional} over a specific support.
\begin{lem}\label{lemma:boundestimate}
Suppose that the matrix $\mathbf{X}_\star$ is the unique minimizer of \eqref{MMV-functional} over all
matrices with support inside $\Lambda$. Then, the following estimates hold:
\begin{eqnarray}
&&\quad \|\mathbf{X}_{0\Lambda}-\mathbf{X}_\star\|_{2\rightarrow\infty}\le\lambda\|(\vect\Gc_{\Lambda}^\ast\vect\Gc_{\Lambda})^{-1}\|_{2\rightarrow\infty}, \label{bound1}\\
&&\quad \|\vect\Gc_\Lambda(\mathbf{X}_{0\Lambda}-\mathbf{X}_\star)\|_F\le\lambda\|\vect\Gc_\Lambda^\dag\|_{2\rightarrow1} \, .\label{bound2}
\end{eqnarray}
\end{lem}
The above results are on the bounds of the error between $\mathbf{X}_\star$ and the ``true" solution $\mathbf{X}_{0\Lambda}$
when the search is restricted to a given support $\Lambda$. 
We now give a condition under which the solution $\mathbf{X}_\star$ to \eqref{MMV-functional} will be supported on $\Lambda$.
For this condition, we need to use the {\em Exact Recovery Coefficient}
\begin{equation}\label{eq:ERC}
ERC(\Lambda)=1-\max_{j\not\in\Lambda}\|\vect\Gc_\Lambda^\dag\vect\wg(\vect y_j)\|_{\ell_1},
\end{equation}
introduced in \cite{TROPP04}, which measures the orthogonality between the column vectors used in $\vect\Gc_\Lambda$ and the remaining column vectors. 

\begin{lem}\label{lemma:ERCbound}
Under the same condition as in Lemma~\ref{lemma:boundestimate}, if the following condition holds
\begin{equation}\label{eq:ERCcondition}
\|\vect\Gc^\ast(\mathbf{B}-\mathbf{B}_\Lambda)\|_{2\rightarrow\infty}\le\lambda ERC(\Lambda),
\end{equation}
then the unique minimizer $\mathbf{X}_\star$ of \eqref{MMV-functional} is supported on $\Lambda$.
\end{lem}
\begin{proof}
By definition, $\mathbf{B}_\Lambda=\vect\Gc_\Lambda\mathbf{X}_{0\Lambda}$. Given any vector $\vect u\in\mC^\nu$,
we have for any $j\not\in\Lambda$,
\begin{equation}\label{eq:innerproduct1}
|\langle\vect\wg^\ast(\vect y_j)(\mathbf{B}-\vect\Gc_\Lambda\mathbf{X}_{0\Lambda}),\vect u\rangle|=|\langle\vect\wg^\ast(\vect y_j)(\mathbf{B}-\mathbf{B}_\Lambda),\vect u\rangle|\le\|(\mathbf{B}-\mathbf{B}_{\Lambda})^\ast\vect\wg(\vect y_j)\|_{\ell_2}\|\vect u\|_{\ell_2}
\end{equation}
and
\begin{equation}\label{eq:innerproduct2}
|\langle\vect\wg^\ast(\vect y_j)\vect\Gc_\Lambda(\mathbf{X}_{0\Lambda}-\mathbf{X}_\star),\vect u\rangle|=|\langle\vect\wg^\ast(\vect y_j)\vect\Gc_\Lambda(\vect\Gc_\Lambda^\ast\vect\Gc_\Lambda)^{-1}\mathbf{G},\vect u\rangle|\le\|\mathbf{G}^\ast\vect\Gc_\Lambda^\dag\vect\wg(\vect y_j)\|_{\ell_2}\|\vect u\|_{\ell_2}.
\end{equation}
Since $\mathbf{X}_\star$ is the unique minimizer among all set of matrices with support included in $\Lambda$, we only
need to show that it is also the optimal solution among matrices with support larger than $\Lambda$.
Let $\vect\zeta\in\mC^K$ be a standard unit vector with support on $\{1,\ldots,K\}\backslash\Lambda$.
Then, $\mathbf{X}_\star+\vect\zeta\vect u^\ast$ is a perturbation by adding a
matrix with row support disjoint from that of $\mathbf{X}_\star$. If we compute the variation 
of \eqref{MMV-functional} with respect to this perturbation, we obtain
\begin{eqnarray*}
L(\mathbf{X}_\star+\vect\zeta\vect u^\ast,\lambda)-L(\mathbf{X}_\star,\lambda)&=&\frac{1}{2}(\|\mathbf{B}-\vect\Gc\mathbf{X}_\star-\vect\wg(\vect y_j)\vect u^\ast\|_F^2-\|\mathbf{B}-\vect\Gc\mathbf{X}_\star\|_F^2)+\\
&&\lambda(J_{2,1}(\mathbf{X}_\star+\vect\zeta\vect u^\ast)-J_{2,1}(\mathbf{X}_\star))\\
&=&\frac{1}{2}\|\vect\wg(\vect y_j)\vect u^\ast\|_F^2-\operatorname{Re}\langle\mathbf{B}-\vect\Gc\mathbf{X}_\star,\vect\wg(\vect y_j)\vect u^\ast\rangle+\lambda\|\vect u\|_{\ell_2}\\
&=&\frac{1}{2}\|\vect\wg(\vect y_j)\vect u^\ast\|_F^2-\operatorname{Re}\langle\vect\wg^\ast(\vect y_j)(\mathbf{B}-\vect\Gc_\Lambda\mathbf{X}_{0,\Lambda}),\vect u^\ast\rangle-\\
&&\operatorname{Re}\langle\vect\wg^\ast(\vect y_j)\vect\Gc_\Lambda(\mathbf{X}_{0,\Lambda}-\mathbf{X}_\star),\vect u^\ast\rangle+\lambda\|\vect u\|_{\ell_2}\\
&>&\lambda\|\vect u\|_{\ell_2}-|\langle\vect\wg^\ast(\vect y_j)(\mathbf{B}-\vect\Gc_\Lambda\mathbf{X}_{0,\Lambda}),\vect u^\ast\rangle|-|\langle\vect\wg^\ast(\vect y_j)\vect\Gc_\Lambda(\mathbf{X}_{0,\Lambda}-\mathbf{X}_\star),\vect u^\ast\rangle|\\
&\ge&\|\vect u\|_{\ell_2}\left(\lambda-\|(\mathbf{B}-\mathbf{B}_{\Lambda})^\ast\vect\wg(\vect y_j)\|_{\ell_2}-\lambda\|\mathbf{G}^\ast\vect\Gc_\Lambda^\dag\vect\wg(\vect y_j)\|_{\ell_2}\right).
\end{eqnarray*}
To show that $L(\mathbf{X}_\star+\vect\zeta\vect u^\ast,\lambda)-L(\mathbf{X}_\star,\lambda)>0$, first observe that condition \eqref{eq:ERCcondition} implies that
\[\|(\mathbf{B}-\mathbf{B}_\Lambda)^\ast\vect\wg(\vect y_j)\|_{\ell_2}
\le 
\|(\mathbf{B}-\mathbf{B}_\Lambda)^\ast\vect\Gc\|_{1\rightarrow2}
=\|\vect\Gc^\ast(\mathbf{B}-\mathbf{B}_\Lambda)\|_{2\rightarrow\infty}
\le
\lambda ERC(\Lambda), \]
and, at the same time, by the definition of $ERC(\Lambda)$ and using Corollary~\ref{col:subgradientnorm}, we obtain
\[\lambda ERC(\Lambda)\le\lambda(1-\|\vect\Gc_\Lambda^\dag\vect\wg(\vect y_j)\|_{\ell_1})
\le
\lambda(1-\|\vect\Gc_\Lambda^\dag\vect\wg(\vect y_j)\|_{\ell_1}\|\mathbf{G}\|_{2\rightarrow\infty})
\le
\lambda(1-\|\mathbf{G}^\ast\vect\Gc_\Lambda^\dag\vect\wg(\vect y_j)\|_{\ell_2}).
\]
Therefore, $L(\mathbf{X}_\star+\vect\zeta\vect u^\ast,\lambda)>L(\mathbf{X}_\star,\lambda)$ which completes the proof.
\end{proof}

With all the supportive results, we are now ready to prove our main result of MMV problem \eqref{MMV21noise}.
\begin{proof}[Proof of Theorem~\ref{thm.mmvnoise}]
Let the support of the solution to \eqref{MMV21noise}, $\mathbf{X}_0$, be $\Lambda_0$ with 
$|\Lambda_0|=M$. We denote the solution by  $\mathbf{X}_{\Lambda_0}$, and the corresponding synthesized data matrix by
$\mathbf{B}_{\Lambda_0}=\vect\Gc\mathbf{X}_{\Lambda_0}$. Since \eqref{MMV21noise} is convex,
the necessary and sufficient condition for it to have a unique solution is that there exists a pair
$(\mathbf{X}_\star,\lambda_\star)$ such that the following KKT conditions are satisfied:
\begin{equation}\label{eq:kkt1}
\mathbf{X}_\star=\argmin_{\mathbf{X}}L(\mathbf{X},\lambda)=\frac{1}{2}\|\mathbf{B}-\vect\Gc\mathbf{X}\|_F^2+\lambda_\star J_{2,1}(\mathbf{X}),
\end{equation}
\begin{equation}\label{eq:kkt2}
\|\mathbf{B}-\vect\Gc\mathbf{X}_\star\|_F=\delta,
\end{equation}
\begin{equation}\label{eq:kkt3}
\lambda_\star>0.
\end{equation}
We first consider the following problem with additional requirement that the support is included in $\Lambda_0$
\begin{equation}\label{eq:supportconstraint}
\min_{\operatorname{rowsupp}(\mathbf{X})\subset\Lambda_0}J_{2,1}(\mathbf{X})\quad\text{s.t.}\,\,\,\|\mathbf{B}-\vect\Gc\mathbf{X}\|_F\le\delta.
\end{equation}
Because $\mathbf{B}_{\Lambda_0}$ is the best Frobenius norm approximation of $\mathbf{B}$,
using Lemma~\ref{lemma:boundestimate} we obtain
\begin{eqnarray*}
\delta^2&=&\|\mathbf{B}-\vect\Gc\mathbf{X}_\star\|_F^2\\
&=&\|\mathbf{B}-\mathbf{B}_{\Lambda_0}\|_F^2+\|\mathbf{B}_{\Lambda_0}-\vect\Gc\mathbf{X}_\star\|_F^2\\
&=&\|\mathbf{B}-\mathbf{B}_{\Lambda_0}\|_F^2+\|\vect\Gc(\mathbf{X}_{\Lambda_0}-\mathbf{X}_\star)\|_F^2\\
&\le&\|\mathbf{B}-\mathbf{B}_{\Lambda_0}\|_F^2+\lambda_\star^2\|\vect\Gc_{\Lambda_0}^\dag\|^2_{2\rightarrow1}.
\end{eqnarray*}
Thus, the second KKT condition \eqref{eq:kkt2} implies that
\[\lambda_\star^2\ge\frac{\delta^2-\|\mathbf{B}-\mathbf{B}_{\Lambda_0}\|_F^2}{\|\vect\Gc_{\Lambda_0}^\dag\|^2_{2\rightarrow1}}.\]
On the other hand, according to Lemma~\ref{lemma:ERCbound}, $\mathbf{X}_\star$ has support on $\Lambda_0$ if
\[\lambda_\star\ge\frac{\|\vect\Gc^\ast(\mathbf{B}-\mathbf{B}_{\Lambda_0})\|_{2\rightarrow\infty}}{ERC(\Lambda_0)}.\]
Therefore, as long as
\[\frac{\delta^2-\|\mathbf{B}-\mathbf{B}_{\Lambda_0}\|_F^2}{\|\vect\Gc_{\Lambda_0}^\dag\|^2_{2\rightarrow1}}\ge\frac{\|\vect\Gc^\ast(\mathbf{B}-\mathbf{B}_{\Lambda_0})\|_{2\rightarrow\infty}^2}{ERC^2(\Lambda_0)},\]
$\mathbf{X}_\star$ is the optimal solution with support included in $\Lambda_0$.
Rearranging the above inequality, we have
\begin{equation}\label{eq:suffcond}
\delta^2\ge\|\mathbf{B}-\mathbf{B}_{\Lambda_0}\|_F^2+\frac{\|\vect\Gc_{\Lambda_0}^\dag\|^2_{2\rightarrow1}\|\vect\Gc^\ast(\mathbf{B}-\mathbf{B}_{\Lambda_0})\|_{2\rightarrow\infty}^2}{ERC^2(\Lambda_0)}.
\end{equation}
By definition,
\[\|\vect\Gc^\ast(\mathbf{B}-\mathbf{B}_{\Lambda_0})\|^2_{2\rightarrow\infty}=\bigg(\max_{1\le j\le K}\|\vect\wg^\ast(\vect y_j)(\mathbf{B}-\mathbf{B}_{\Lambda_0})\|_{\ell_2}\bigg)^2\le
\|\mathbf{B}-\mathbf{B}_{\Lambda_0}\|_F^2.
\]
According to Propositions~$3.7$ and $3.9$ in \cite{TROPP04},
\[\frac{\|\vect\Gc_{\Lambda_0}^\dag\|_{2\rightarrow1}^2}{ERC^2(\Lambda_0)}\le\frac{M(1-(M-1)\epsilon)}{(1-2M\epsilon+\epsilon)^2}.\]
Hence, we have
\[\|\vect{\Ec}\|_F^2\left(1+\frac{M(1-(M-1)\epsilon)}{(1-2M\epsilon+\epsilon)^2}\right)\ge\|\mathbf{B}-\mathbf{B}_{\Lambda_0}\|_F^2+\frac{\|\vect\Gc_{\Lambda_0}^\dag\|^2_{2\rightarrow1}\|\vect\Gc^\ast(\mathbf{B}-\mathbf{B}_{\Lambda_0})\|^2_{2\rightarrow\infty}}{ERC^2(\Lambda_0)}.\]
Therefore, condition \eqref{eq:constraintcondition} is sufficient for \eqref{eq:suffcond} to hold and $\mathbf{X}_\star$ is the unique minimizer to \eqref{MMV21noise} with support inside $\Lambda_0$.

Next we show that this minimizer over the support $\Lambda_0$ is also the global minimizer to \eqref{MMV21noise}.
Assume there exists another coefficient matrix $\widehat{\mathbf{X}}$ which minimizes
\eqref{MMV21noise} and thus also satisfies the KKT conditions, especially \eqref{eq:kkt2}.
Then $\vect\Gc\mathbf{X}_\star=\vect\Gc\widehat{\mathbf{X}}$ must hold. Assume this is not the case. Then
since formulation \eqref{MMV21noise} is convex, any linear combination of solutions will also be a solution.
In particular, $\frac{1}{2}(\mathbf{X}_\star+\widehat{\mathbf{X}})$ is a solution and should satisfies KKT
condition \eqref{eq:kkt2}. This is a contradiction because
$$\|\mathbf{B}-\frac{1}{2}\vect\Gc\mathbf{X}_\star-\frac{1}{2}\vect\Gc\widehat{\mathbf{X}}\|_F<\delta.$$
Now that both $\mathbf{X}_\star$ and $\widehat{\mathbf{X}}$ minimize \eqref{MMV21noise} with the same
value $\vect\Gc\mathbf{X}_\star$. It implies that both solutions satisfy
\[\min_{\mathbf{X}}J_{2,1}(\mathbf{X})\quad\text{s.t.}\,\,\,\vect\Gc\mathbf{X}=\vect\Gc\mathbf{X}_\star.\]
However, due to Theorem~\ref{thm.mmv}, when $M\epsilon<1/2$, the above optimization has a unique solution. We
then prove that $\mathbf{X}_\star=\widehat{\mathbf{X}}$, i.e. the solution to \eqref{MMV21noise} is unique.

Finally, the error bound of the minimizer compared to underlying solution is estimated as follows
\[\|\mathbf{X}_\star-\mathbf{X}_0\|_F=\|(\vect\Gc^\ast\vect\Gc)^{-1}\vect\Gc^\ast\vect\Gc(\mathbf{X}_\star-\mathbf{X}_0)\|_F\le\|\vect\Gc_{\Lambda_0}^\dag\|_{2\rightarrow2}\|\vect\Gc(\mathbf{X}_\star-\mathbf{X}_0)\|_F\le\delta/\sqrt{1-(M-1)\epsilon},\]
where we use the singular value estimate of $\vect\Gc_{\Lambda_0}$ given in \cite{DET06} and \cite{TROPP04}. Note that
if $\|(\mathbf{X}_0)_{i\cdot}\|_{\ell_2}>\delta/\sqrt{1-(M-1)\epsilon}$ for a row $i$, then  $\|(\mathbf{X}_\star)_{i\cdot}\|_{\ell_2}$ cannot be $0$ and, therefore, component $i$ is included in the recovered support.
\end{proof}

\section{Proof of results in \S\ref{sec:arrayconfig}}\label{appendix:innerproduct}
In this section, we will use $\theta$ for azimuthal angle, $\phi$ for polar angle
and $\Omega$ for the area of imaging array. We also assume the size of the array
$a$ is much larger than the distance $h$ between any two neighboring transducers.
\begin{proof}[Proof of Proposition~\ref{prop:spherical array}]
For spherical arrays of radius $L$, given any point $\vect x$ on the array and
$\vect y$ in IW, we have $|\vect x-\vect y|\approx L$. 
With the continuum approximation
\begin{equation*}
\|\vect\wg_0(\vect y)\|^2_{\ell_2}=\sum_{\vect x}\left|\frac{\exp(-\mathrm{i}\kappa|\vect x-\vect y|)}{4\pi|\vect x-\vect y|}\right|^2
\approx\frac{1}{16\pi^2h^2}\int_{\Omega}\frac{\mathrm{d}\vect x}{|\vect x-\vect y|^2}=\frac{1}{16\pi^2h^2L^2}\times(4\pi L^2)=\frac{1}{4\pi h^2},
\end{equation*}
i.e. the norm of Green's function vector is constant under the spherical array.
On the other hand, using continuum approximation, we have for the inner product
of any two Green's function vector at $\vect y_k$ and $\vect y_{k'}$,
\[\vect\wg_0^\ast(\vect y_k)\vect\wg_0(\vect y_{k'})\approx\frac{1}{16\pi^2h^2}\int_{\Omega}\frac{\exp\big(\mathrm{i}\kappa(|\vect x-\vect y_{k'}|-|\vect x-\vect y_k|)\big)}{|\vect x-\vect y_{k'}||\vect x-\vect y_k|}\,\mathrm{d}\vect x,\]
where the integral is taken on the sphere of radius $L$, i.e. $\Omega=\{\vect x: |\vect x|=L\}$.
Let $\widehat{\vect x}=\frac{\vect x}{L}$ so $|\widehat{\vect x}| = 1$ on the integral area. Because $|\vect y|\ll L$, we have the approximation
\[|\vect x-\vect y| = L|\widehat{\vect x}-\frac{\vect y}{L}| = L \sqrt{|\widehat{\vect x}|^2 + \frac{|\vect y|^2}{L^2} - 2 \widehat{\vect x}\cdot\frac{|\vect y|}{L}} \approx L-\widehat{\vect x}^\ast\vect y\, , \]
and therefore
\[|\vect x-\vect y_{k'}|-|\vect x-\vect y_k|\approx\widehat{\vect x}^\ast(\vect y_k - \vect y_{k'}).\]
Using these approximations, and since $|\widehat{\vect x}-\vect y_{k'}/L| \approx |\widehat{\vect x}-\vect y_k/L| \approx 1$, we have
\begin{eqnarray*}
\vect\wg_0^\ast(\vect y_k)\vect\wg_0(\vect y_{k'})&\approx&\frac{1}{16\pi^2h^2}\int_{|\widehat{\vect x}|=1}\frac{\exp\big(\mathrm{i}\kappa\widehat{\vect x}\cdot(\vect y_k-\vect y_{k'})\big)}{|\widehat{\vect x}-(\vect y_{k'}/L)||\widehat{\vect x}-(\vect y_k/L)|}\,\mathrm{d}\widehat{\vect x}\\
&\approx&\frac{1}{16\pi^2h^2}\int_0^{2\pi}\,\mathrm{d}\theta\int_0^\pi\,\exp(\mathrm{i}\kappa|\vect y_{k}-\vect y_{k'}|\cos\phi) \sin\phi \,\mathrm{d}\phi \\
&=&\frac{1}{8\pi h^2}\int_0^\pi\exp(\mathrm{i}\kappa|\vect y_{k}-\vect y_{k'}|\cos\phi)\sin\phi\,\mathrm{d}\phi\\
&=&\frac{1}{4\pi h^2}\frac{\sin\kappa|\vect y_k-\vect y_{k'}|}{\kappa|\vect y_k-\vect y_{k'}|} = \frac{1}{4\pi h^2} \sinc(\kappa|\vect y_k-\vect y_{k'}|) ,
\end{eqnarray*}
where we changed the surface integral to an integral characterized by the angles $\theta$ and $\phi$, with $\phi$ the angle between 
$\vect y_k-\vect y_{k'}$ and $\widehat{\vect x}$. Using the approximate form of the norm of $\vect\wg_0(\vect y)$, we have
\[\frac{\vect\wg_0^\ast(\vect y_k)\vect\wg_0(\vect y_{k'})}{\|\vect\wg_0(\vect y_k)\|_{\ell_2}\|\vect\wg_0(\vect y_{k'})\|_{\ell_2}}\approx\sinc(\kappa|\vect y_k-\vect y_{k'}|).\]
\end{proof}

\begin{proof}[Proof of Proposition~\ref{prop:planar array}]
We first calculate the norm of Green's function vector under the planar array as follows
\begin{equation}
\|\vect\wg_0(\vect y)\|^2_{\ell_2}
\approx\frac{1}{16\pi^2h^2}\int_{\Omega}\frac{\mathrm{d}\vect x}{|\vect x-\vect y|^2}=\frac{1}{16\pi^2h^2}\int_0^{2\pi}\,\mathrm{d}\theta\int_0^{\phi_0}\tan\phi\,\mathrm{d}\phi=-\frac{1}{8\pi h^2}\log(\cos\phi_0),
\end{equation}
where $\phi_0=\arctan(\frac{a}{2L})$ is the maximal polar angle determined by the size $a$ of the imaging array and the distance $L$ from the array to the IW. Using the identity $\cos(\arctan(x)) = 1 / \sqrt{ 1 + x^2}$, we obtain
\begin{equation}
\|\vect\wg_0(\vect y)\|^2_{\ell_2}
\approx
\frac{1}{16\pi h^2}\log\bigg(1+\frac{a^2}{4L^2}\bigg).
\end{equation}
Hence, for planar arrays, the norm depends on $a$ and $L$ and is independent of the pixel size of the IW.

Based on the proof of Proposition~$3.1$ in \cite{CMP13},
when $\vect y_k - \vect y_{k'}\perp\vect y_k$, it can be seen that the inner product
$$|\vect\wg_0^\ast(\vect y_{k'})\vect\wg_0(\vect y_k)|\sim1/\sqrt{\kappa|\vect y_k-\vect y_{k'}|}.$$
Therefore, we only need to show below when $\vect y_k-\vect y_{k'}\parallel\vect y_k$, the inner prodcut decays no worse than
$1/\sqrt{\kappa|\vect y_k-\vect y_{k'}|}$.

According to \cite{CMP13}, when $|\vect y_k-\vect y_{k'}|\ll L$ and $(\vect y_k-\vect y_{k'})\parallel\vect y_k$, we have that
\begin{equation*}
\vect\wg_0^\ast(\vect y_k)\vect\wg_0(\vect y_{k'})\approx\frac{1}{8\pi h^2}\int_{\cos\phi_0}^1\frac{\exp(-\mathrm{i}\kappa\eta z)}{z}\,\mathrm{d}z,
\end{equation*}
where $\eta=|\vect y-\vect y^S|$. When $\kappa\eta\rightarrow\infty$, the integrand oscillates very fast provided that
$1/\kappa\eta\ll \cos\phi_0 \ll 1$. In this case, integration by parts gives the leading asymptotic behavior as $\kappa\eta\rightarrow\infty$. Explicitly,
\[\int_{\cos\phi_0}^1\frac{\exp(-\mathrm{i}\kappa\eta z)}{z}\,\mathrm{d}z=\frac{\mathrm{i}}{\kappa\eta}\bigg(\exp(-\mathrm{i}\kappa\eta)-\frac{\exp(-\mathrm{i}\kappa\eta\cos\phi_0)}{\cos\phi_0}-\int_{\cos\phi_0}^1\frac{\exp(-\mathrm{i}\kappa\eta z)}{z^2}\,\mathrm{d}z\bigg).\]
The integral on the right hand side vanishes more rapidly than the boundary terms as $\kappa\eta\rightarrow\infty$
(to see this, integrate $\int_{\cos\phi_0}^1\frac{\exp(-\mathrm{i}\kappa\eta z)}{z^2}\,\mathrm{d}z$ by parts and notice that it vanishes like $1/\kappa\eta$).
Therefore, neglecting the integral on the right hand side, we obtain
\[
\left|\int_{\cos\phi_0}^1\frac{\exp(-\mathrm{i}\kappa\eta z)}{z}\,\mathrm{d}z\right|
\sim
\frac{1}{\kappa\eta\cos\phi_0}\bigg|
\cos\phi_0 - \exp(-\mathrm{i}\kappa\eta(\cos\phi_0-1)) \bigg|\,\,\text{as}\,\, \kappa\eta\rightarrow\infty\,.
\]
Thus,
\[
\left|\int_{\cos\phi_0}^1\frac{\exp(-\mathrm{i}\kappa\eta z)}{z}\,\mathrm{d}z\right|
\sim
\frac{1}{\kappa\eta\cos\phi_0}\sqrt{\cos^2\phi_0+1-2\cos\phi_0\cos(\kappa\eta(\cos\phi_0 - 1))}\,\,\text{as}\,\,\kappa\eta\rightarrow\infty\,.
\]

For large arrays $a \gg L$, we can approximate $\cos\phi_0=2L/\sqrt{a^2+4L^2}$ by $0$ and obtain
$|\vect\wg_0^\ast(\vect y_k)\vect\wg_0(\vect y_{k'})|\approx 1/(\kappa\eta\cos\phi_0)$
which implies that, for large arrays, the normalized inner product decreases like
$1/(\kappa\eta\cos\phi_0\log(\sec\phi_0))$, as $\kappa\eta\rightarrow\infty$.
This function depends very smoothly respect to $\cos\phi_0$ when $1/\kappa\eta\ll\cos\phi_0\ll1$,
i.e., it is almost independent of $a/L$.

Moreover, we find that
\[\frac{1}{\kappa\eta}\left(\frac{2}{\cos\phi_0}-2\right)\le\left|\int_{\cos\phi_0}^1\frac{\exp(-\mathrm{i}\kappa\eta z)}{z}\,\mathrm{d}z\right|\le\frac{2}{\kappa\eta\cos\phi_0},\]
so we get the following bounds
\[\frac{1}{\kappa\eta\log(\sec\phi_0)}\left(\frac{2}{\cos\phi_0}-2\right)\le\left|\frac{\vect\wg_0^\ast(\vect y_k)\vect\wg_0(\vect y_{k'})}{\|\vect\wg_0(\vect y_k)\|_{\ell_2}\|\vect\wg_0(\vect y_{k'})\|_{\ell_2}}\right|\le\frac{2}{\kappa\eta\cos\phi_0\log(\sec\phi_0)}.\]
Together with the estimate of the cases when $\vect y_k-\vect y_{k'}\perp\vect y_k$, we can see the inner product,
when planar array is used, has decay rate $\frac{1}{\sqrt{\kappa\eta}}$.
\end{proof}

\end{appendices}

\end{document}